\documentclass[letterpaper,11pt]{article}
\usepackage{fullpage}
\usepackage[utf8]{inputenc}
\usepackage[T1]{fontenc}
\usepackage{verbatim}
\usepackage{graphicx}
\usepackage{tabularx}
\usepackage{array}
\newcolumntype{L}{>{\raggedright\arraybackslash}X}
\newcolumntype{R}{>{\raggedleft\arraybackslash}X}
\usepackage{delarray}
\usepackage{dcolumn}
\usepackage{float}
\usepackage{amsmath}
\usepackage{psfrag}
\usepackage{booktabs}
\usepackage{multirow,hhline}
\usepackage{amstext}
\usepackage{amssymb}
\usepackage{textcomp}
\usepackage{fancyhdr}
\usepackage[onehalfspacing]{setspace}
\usepackage{amsthm}
\usepackage[ruled,linesnumbered]{algorithm2e}
\usepackage{makeidx}
\usepackage{bm}
\usepackage{url}
\usepackage{adjustbox}
\usepackage{enumitem}
\usepackage{ragged2e}
\usepackage{booktabs}
\usepackage{titling}
\usepackage{subcaption}
\usepackage{threeparttable}
\usepackage{longtable}
\usepackage{pdflscape}
\usepackage[round,authoryear,longnamesfirst]{natbib}
\usepackage{authblk}
\makeatletter
\renewcommand\AB@authnote[1]{\protect\raisebox{0.6ex}{\protect\scriptsize\protect\normalfont#1}}
\renewcommand\AB@affilnote[1]{\protect\raisebox{0.6ex}{\protect\scriptsize\protect\normalfont#1}}
\makeatother
\usepackage{hyperref}
\usepackage[svgnames]{xcolor}
\hypersetup{
	colorlinks=true,
	citecolor=NavyBlue,
    filecolor=NavyBlue,
    linkcolor=NavyBlue,
    urlcolor=NavyBlue}

\makeindex
\allowdisplaybreaks
\numberwithin{equation}{section}
\newtheorem{lem}{Lemma}[section]
\newtheorem{thm}{Theorem}[section]
\newtheorem{cor}{Corollary}[section]

\newtheorem{ass}{Assumption}
\newtheorem{prop}{Proposition}[section]
\newtheorem{rem}{Remark}

\newcommand{\diag}{\text{diag}}

\DeclareMathOperator*{\argmax}{arg\,max}

\makeatletter
\newcommand*{\rom}
[1]{\expandafter\@slowromancap\romannumeral #1@}
\makeatother
\usepackage{xr}

\usepackage{microtype}

\newtheorem{proposition}[prop]{Proposition}
\newtheorem{lemma}[lem]{Lemma}
\newtheorem{corollary}[cor]{Corollary}

\theoremstyle{definition}

\theoremstyle{remark}
\newtheorem{remark}[rem]{Remark}

\newcommand{\R}{\mathbb R}

\usepackage{placeins}

\begin{document}
\title{Vine Copula VAR:\\ From Recursive Margins to Joint Forecast Inference\thanks{Previous versions of this paper have been circulated under the title ``Vine Copula VAR". We thank seminar participants at the 2026 Midwest Econometrics Group Meeting and the CUNY Graduate Center for their valuable comments. Yubo Tao acknowledges the financial support from the National Natural Science Foundation of China (Project No.~72303003) and a seed grant from the Asia-Pacific Academy of Economics and Management of the University of Macau. }}
\author[a]{Hunter Ng}
\affil[a]{\small\emph{Zicklin School of Business, Baruch College, City University of New York, USA}}
\author[b]{Yubo Tao}
\affil[b]{\small\emph{Faculty of Social Sciences, University of Macau, Macao SAR, China}}
\date{October 2026}
\maketitle
\begin{abstract}
\noindent
Joint-event forecasts combine an estimated dependence model with current
marginal forecasts. Retaining historical transforms at the marginal fits
available when each outcome arrived creates overlapping estimation
errors. We derive their joint influence with the terminal estimates in
a stable Vine Copula VAR with normal innovation margins and a fixed,
correctly specified Gaussian or positive Clayton vine. An intercept
identity and the stable VAR filter reduce the historical correction to
harmonically weighted innovation moments. Terminal slope uncertainty
remains in the forecast gradient. A feasible covariance estimator yields
asymptotically valid repeated-sample intervals for fixed one-sided event
probabilities at the realized forecast state. A same-family terminal
refit provides the efficiency comparison. Retaining dated transforms
adds a positive semidefinite limiting covariance term, whose event-level
cost depends on dependence sensitivity and forecast-state leverage.
Monte Carlo comparisons show a small accuracy cost of retaining the
transforms and a larger inferential role for current marginal estimation.
Allowing Clayton dependence improves downside forecasts relative to
Gaussian dependence under each transform rule. A fourteen-variable
U.S. real-time forecasting application and simulations with twenty and
fifty variables illustrate feasible joint-event estimation and analytic
inference for fixed-depth vines.
\par\medskip
\noindent\textbf{Keywords}: Vine copulas, Recursive estimation, Generated regressors, Joint-event forecasts, Tail dependence.
\par\medskip
\noindent\textbf{JEL Classification}: C32, C53, C58.
\end{abstract}
\clearpage
\section{Introduction}\label{sec:introduction}

A forecaster assessing joint contractions in employment and production
must combine current marginal forecasts with an estimate of dependence.
The historical inputs to that estimate involve a choice. An outcome can
be measured against the marginal forecast made before it arrived, or
against a marginal distribution re-estimated using the complete current
sample. The first rule preserves a record of successive forecast
surprises. The second revises the historical residuals whenever the
marginal model is updated. We study how this choice affects inference
for the next conditional event probability.

The distinction matters even when both rules use the same observations.
A fall in employment and production can be largely anticipated at one
date and unusual at another. Its positions within the two preceding
marginal forecast distributions describe the joint surprise at that
date. A copula fitted to these positions estimates dependence from the
historical marginal forecasting rule. The next forecast combines that
estimate with the current margins. Repeated re-estimation makes the
historical positions statistically dependent through their overlapping
training samples, even when the underlying innovation vectors are
independent over time. Treating the marginal inputs as known therefore
omits part of the uncertainty in the estimated event probability.

Preserving the forecast record is motivated by the sequential assessment
of issued forecasts in prequential analysis \citep{dawid_1984} and
the probability integral transform (PIT) approach to density-forecast
evaluation \citep{diebold_gunther_tay_1998}. Copulas also provide a way
to combine existing marginal forecasts. For example,
\citet{mogliani_odendahl_2025} estimate cross-horizon dependence from
historical forecast PITs when transforming supplied density forecasts.
When issued marginal forecasts are the inputs to a combination procedure,
retaining their PITs keeps dependence estimation tied to the original
forecast errors. The dependence stage can then be updated while preserving
the marginal forecast record used for evaluation.
Our problem concerns dependence across variables after recursive VAR
estimation. We take retention of the dated marginal transforms as the
specified forecasting rule and derive its inferential consequences.
The application reconstructs this rule from real-time observations.

A terminal refit remains feasible because the raw history is available.
It uses information available at the current forecast origin and is a
legitimate alternative for forecasting the next outcome. Under our
stationary model, the two procedures target the same innovation copula
and conditional event probability. The refit supplies a useful efficiency
comparison: we show that retaining dated transforms adds a positive
semidefinite term to the limiting covariance. The result quantifies the
cost of preserving the marginal forecast record and identifies the
correction needed for inference under that rule.

Vines make the dependence component tractable. They build a multivariate
copula from bivariate copulas for pairs and conditional pairs
\citep{aas_pair-copula_2009,nagler_kruger_min_2022}. Gaussian pairs
describe symmetric dependence, while positive Clayton pairs allow
lower-tail dependence \citep{czado_nagler_2022}. Normal coordinate
margins can therefore coexist with an asymmetric joint distribution.
We cover both families in a stable VAR with constant innovation scales
and a fixed sequential vine. The distinction between copula families
and the distinction between dated and refitted transforms answer
separate questions. Our benchmark comparisons vary one at a time.

Marginal equations are estimated on expanding samples, the vine is fitted
sequentially to the retained
PITs, and the next event probability combines the vine estimate with the
terminal VAR estimates. The general need to propagate first-stage error
is familiar from two-step econometric inference
\citep{murphy_topel_1985} and copula estimation with fitted margins
\citep{joe_2005,patton_2012}. Here the first stage is an entire sequence
of overlapping marginal fits. An observation affects many subsequent
PITs, and the resulting errors also pass through the conditional
transforms used in later vine trees.

The central step is to identify which historical estimation errors
survive averaging. An intercept identity and the stable VAR filter
cancel the centered lag-regressor contribution, leaving harmonic weights
on innovation location and scale moments. Earlier observations receive
larger weights because they enter more subsequent regressions. The
terminal forecast is evaluated at a single lag state, so every terminal
slope remains in its gradient. We derive the joint covariance of the
sequential vine estimate, all terminal coefficients, and the terminal
scales, including their cross covariances. Its feasible estimate is
positive semidefinite and uses the retained forecast record.

Inference also requires care because the terminal state remains random
as the sample grows. We separate its recent innovations from the leading
estimation error and obtain a state-dependent mixed normal limit.
Studentization then yields a normal pivot. The intervals have
repeated-sample coverage for the conditional probability evaluated at
each sample's realized state. The theory is pointwise under a fixed
stationary law, with fixed dimension and lag order, iid innovation
vectors, normal coordinate margins, and a correctly specified simplified
vine. We verify the score and event conditions for Gaussian and positive
Clayton pairs, including their estimated ancestors. Fixed one-sided
events have positive limiting variance almost surely, including mixed
directions and subsets of variables. Bounded rectangles require an
event-specific nondegeneracy condition.

A same-family terminal refit permits a direct comparison on the same
observations. The comparison applies to Gaussian and positive Clayton
vines. The two estimators have identical limiting
terminal-margin and cross covariance blocks. Retaining issued transforms
adds a positive semidefinite term to the copula block, giving the archive's
variance cost for any covered event. In a bivariate specialization,
this term arises entirely from scale sensitivity and vanishes at
independence. Its share of event variance also depends on uncertainty
in the terminal conditional mean. A fixed-event submodel isolates this
denominator effect: increasing lag-state leverage reduces the archive
share while leaving its added variance unchanged. The comparison holds
the observations fixed. Replacing them with revised data changes the
forecasting problem. We also retain exact finite training weights in
the covariance estimate and justify intervals on the log-odds scale.
Both modifications preserve the first-order coverage result.

Existing vine methods cover both conditional distributions and quantiles.
\citet{kraus_czado_2017} develop D-vine quantile regression, while
\citet{chang_joe_2019} compute conditional distributions under
regular-vine models with continuous and discrete variables. For
multivariate time series, \citet{smith_2015} uses a D-vine to represent
both nonlinear serial and cross-sectional dependence. Our question
concerns inference when the historical transforms retain
the marginal estimates available at their issue dates.

The closest results on inference for prediction are those of
\citet{nagler_kruger_min_2022}. Their stationary-vine framework uses
pair copulas to represent both serial and cross-sectional dependence,
extending copula time-series models such as
\citet{remillard_papageorgiou_soustra_2012}. Our vine describes
contemporaneous innovations after the VAR accounts for linear serial
dynamics. The inferential distinction concerns the marginal inputs.
Section~4.1 of
\citet{nagler_kruger_min_2022} forms historical transforms using a
common marginal estimate, and their Theorems~8--9 propagate marginal
and copula estimation error to prediction functionals. We calculate
the accumulated influence of dated marginal estimates and its joint
law with the terminal forecast parameters. The recursive weighting is
related to \citet{west_1996},
while the triangular score system builds on sequential vine estimation
\citep{hobaeck_haff_2013}. The sample-separation principle of
\citet{beutner_heinemann_smeekes_2021} addresses the additional issue
posed by a random conditioning state. We establish the required
separation for the recursive influence derived here.

Observation-specific marginal inputs also appear in
\citet[Section 2.3]{sancetta_2008}, who allows surrogate margins,
including a prequential construction, and studies posterior consistency
in an elliptic copula model. We derive the joint first-order law of
the vine fit based on dated marginal estimates, terminal VAR estimates, and forecast state.
That law supports feasible event inference and a same-family covariance
comparison on the same observations.

Other copula time-series models address different sources of marginal
and dynamic uncertainty. \citet{chen_xiao_wang_2022} study filtered
nonstationarity, and \citet{chen_wang_xiao_yi_2025} develop a pseudo
sieve likelihood approach that jointly estimates the residual copula
parameter and invariant density. \citet{zhao_shi_zhang_2022} construct
conditional distributions through copula-linked univariate D-vines,
while \citet{tsionas_izzeldin_trapani_2022} allow time-varying VAR
parameters. Our fixed OLS and residual-variance rule is more restrictive
but makes the historical correction and its terminal cross covariances
explicit. Dimension remains fixed throughout. The common-fit,
increasing-dimension results of \citet{gauss_nagler_2026} concern a
different asymptotic regime.

Monte Carlo simulations separate the effects of transform timing,
copula-family specification, and estimation uncertainty. For the two
principal events, the proposed logit intervals have coverage of
94.1\%--95.5\% in the baseline designs. Accounting for current marginal
estimation improves coverage much more than the additional historical
variance term. Same-family terminal refits generally have slightly
lower probability RMSE and shorter intervals than the dated-transform
estimator, consistent with the covariance comparison. Clayton fits
outperform Gaussian fits under Clayton innovations with either dated
or refitted transforms. This accuracy gain comes from the dependence
family. Comparisons with the stationary-vine method of
\citet{nagler_kruger_min_2022} use a common corrected point fit for both
its native bootstrap and a full-Jacobian bootstrap. Their coverage
rankings depend on persistence and the bootstrap construction, so the
numerical evidence supports inference for the specified forecasting
rule without a general superiority claim.

A real-time forecasting application uses fourteen macroeconomic variables to
forecast joint U.S. employment and production contractions at 231
monthly origins. The broader information set includes labor input,
capacity pressure, prices, monetary conditions, and credit. Nested
three- and eleven-variable systems use identical issue dates and locked
release outcomes, separating the dimension comparison from data revisions
\citep{croushore_stark_2001,koenig_dolmas_piger_2003}. Before 2020,
employment--production Brier loss falls from $0.0495$ with three
variables to $0.0446$ with fourteen, although the comparison narrowly
misses the 5\% level after multiplicity adjustment. The fourteen-variable
archive is competitive with Gaussian and stationary-vine benchmarks,
but independence has lower observed loss in both evaluation periods
and the sparse stationary vine performs better during 2020--2024.
The fourteen-variable model improves the observed pre-2020 score
relative to the three-variable model, while the benchmark comparisons
show which alternatives forecast each evaluation period more accurately.

The same application illustrates a practical computational benefit.
A fourteen-variable recurring probability update takes a median
$0.0027$ seconds, followed by $0.0123$ seconds for analytic inference.
Corrected stationary-vine point fits take about $0.36$--$0.37$ seconds,
while the dense Gaussian VAR is faster at point forecasting. Additional
twenty- and fifty-variable simulations assess accuracy, interval
calibration, and computing cost at longer sample sizes. This evidence
supports computational scalability for fixed-depth vines and specified
low-dimensional event targets. Current marginal estimation also
accounts for much of the uncertainty in the application, connecting
its analytic intervals to the paper's joint covariance formula.

Section~\ref{sec:method} defines the forecast rule.
Section~\ref{sec:theory} derives the joint limit and event intervals,
and Sections~\ref{sec:simulation} and~\ref{sec:evaluation} present the
simulations and real-time application. The main appendices contain the
central proofs. Supporting lemmas, implementation details, and
additional evidence appear in the online supplement.

\section{Model and Estimation}\label{sec:method}\label{sec:model}

The VAR describes serial dynamics, and the vine describes
contemporaneous innovation dependence, following the marginal and
copula approach reviewed by \citet{patton_2012}. Inference must account
for the timing of marginal estimation: historical PITs use dated fits,
whereas the event forecast uses terminal estimates. We define the
target, both estimation stages, and the record needed for inference.

\subsection{A stationary VAR and a conditional event}

For fixed $n,p\ge1$, let $y_t=(y_{1t},\ldots,y_{nt})'$ follow the causal
stationary solution of
\begin{equation}\label{main:var-model}
 y_t=a+\sum_{j=1}^p A_jy_{t-j}+D_\sigma Z_t,\qquad
 D_\sigma=\diag(\sigma_1,\ldots,\sigma_n),\quad \sigma_i>0.
\end{equation}
\begin{ass}[Stable VAR and innovations]\label{main:var-assumption}
The dimensions $n,p$ and parameters in \eqref{main:var-model} are fixed.
The lag companion matrix $\mathcal A$ and innovations satisfy
\[
 \rho(\mathcal A)<1,\qquad Z_t\overset{\mathrm{iid}}{\sim}P_Z,
 \qquad Z_{it}\sim N(0,1),\qquad R_Z=E(Z_tZ_t')\succ0.
\]
The process is its causal stationary solution, $\sigma_i>0$ for every
$i$, and the initial $p$ lag vectors are observed.
\end{ass}
Innovation vectors are independent across dates, while their coordinates
may be dependent. Thus Gaussian coordinate margins permit non-Gaussian
joint dependence, which the vine will specify. Write
$\mathcal F_t=\sigma(Z_s:s\le t)$. Stability makes the history of $y_t$
measurable with respect to this field. Define
\begin{equation}\label{eq:design}
 L_t=(y_{t-1}',\ldots,y_{t-p}')',\qquad x_t=(1,L_t')',\qquad
 y_{it}=x_t'\beta_i+\sigma_iZ_{it}.
\end{equation}
The vector $\beta_i$ contains an intercept and $np$ lag coefficients,
giving dimension $d_x=1+np$. The stable VAR expresses $L_t$ as a
geometrically decaying function of past innovations. The current
innovation is therefore independent of the design $x_t$, as required
by the OLS influence calculation. Parameters remain constant, while
forecasts vary with the observed lag state.

For a fixed deterministic threshold vector $c\in\mathbb R^n$, the
one-step forecast target is
\begin{equation}\label{main:conditional-target}
 p_{T+1}(c)=P(y_{T+1}\le c\mid\mathcal F_T)
 =C_{\theta_0}\left(\Phi\left\{
        \frac{c_i-x_{T+1}'\beta_i}{\sigma_i}\right\}:i\le n\right).
\end{equation}
The inequalities hold coordinatewise. Although $c$ is fixed,
$p_{T+1}(c)$ is random across samples because the terminal state
$x_{T+1}$ is random. This distinction determines the coverage statement
in Section~\ref{main:forecast-section}. We first consider the lower
orthant in \eqref{main:conditional-target}, then extend inference to
fixed rectangles, mixed directions, and unrestricted coordinates. The
labor-market event in Section~\ref{sec:evaluation} uses that extension.
Estimating the threshold itself would introduce another influence term.

\subsection{Recursive margins and the archive}

Using the first $m$ observations, compute the OLS coefficient estimator
and residual variance with the following safeguards:
\begin{align}
 \widehat Q_m&=m^{-1}\sum_{s=1}^m x_sx_s',\notag\\
 \widehat\beta_{i,m}
 &=\begin{cases}
 \widehat Q_m^{-1}m^{-1}\sum_{s=1}^m x_sy_{is},
       &\lambda_{\min}(\widehat Q_m)\ge m^{-1/4},\\
 0,&\text{otherwise},
 \end{cases} \label{dyn:ols}\\
 \widehat v_{i,m}
 &=\max\left\{m^{-1/4},\,
       m^{-1}\sum_{s=1}^m(y_{is}-x_s'\widehat\beta_{i,m})^2\right\},\notag\\
 \widehat Z_{it}&=\frac{y_{it}-x_t'\widehat\beta_{i,t-1}}
                         {\sqrt{\widehat v_{i,t-1}}}.
 \label{dyn:variance}
\end{align}
The Gram-matrix threshold and variance floor keep the estimator defined
even when an early regression is ill-conditioned. They become inactive
on all retained estimation dates simultaneously with probability tending to one
(Lemma~\ref{dyn:first-stage}), but remain part of the finite-sample
procedure. The standardized residual $\widehat Z_{it}$ uses estimates
available only through $t-1$.
The true and fitted conditional margins are
\begin{equation}\label{eq:gaussian-margin}
 F^0_{i,t}(y)=\Phi\left(\frac{y-x_t'\beta_i}{\sigma_i}\right),\qquad
 \widehat F_{i,t|t-1}(y)=
 \Phi\left(\frac{y-x_t'\widehat\beta_{i,t-1}}
                    {\sqrt{\widehat v_{i,t-1}}}\right).
\end{equation}
Thus $y_t$ is evaluated against a margin fitted before it enters the
regression. With $b_T=\lceil T^\eta\rceil$, $0<\eta<1$, retain the archive
\begin{equation}\label{eq:stored-pits}
 \widehat U_{it}=\widehat F_{i,t|t-1}(y_{it}),\qquad b_T<t\le T.
\end{equation}
Once recorded, a transform remains fixed. The current forecast margin
uses $\widehat\beta_{i,T}$ and $\widehat v_{i,T}$.

For each issue date, retain the marginal fit and the regressor vector, then
compute the PIT when the outcome arrives. The mathematical record
needed at origin $T$ is
\[
 \mathcal A_T=\{(t,\widehat U_t,x_t):b_T<t\le T\},\qquad
 \widehat Q_T,\quad
 (\widehat\beta_{i,T},\widehat v_{i,T}:i\le n),\quad x_{T+1}.
\]
The PITs determine the score derivatives, while $x_t$ and
$\widehat Q_T$ determine the regression influences in
\eqref{main:joint-feasible}. Joint inference therefore requires more
than a PIT series. In numerical storage, retain the finite standardized
residuals $\widehat Z_t$ as well. Evaluating $\Phi$ can round an extreme
score to zero or one, after which inverse transformation cannot recover
it. Storing the scores preserves the information used by the estimator
without altering its mathematical definition.

At each date, the order of operations preserves the forecast's
information set. On observing $y_t$, evaluate it using the marginal fit
through $t-1$ and store $\widehat U_t$. Then add $y_t$ to the regression
for the next forecast. At origin $T$, the copula uses this record of
marginal forecast surprises, while the forecast for $T+1$ uses the fit
through $T$. Each marginal training input is thus evaluated out of
sample, and the next forecast uses the full available sample.
Recomputing a stored PIT from the fit through $T$ changes the training
criterion even when the observation itself is unchanged.

The retained PITs need not have exactly uniform finite-sample margins
because each dated fit contains estimation error. The population target
is still the stationary innovation copula, and the joint inference
accounts for the generated inputs. A terminal refit targets the same
copula using in-sample marginal residuals. In the Gaussian submodel,
refitting weakly reduces asymptotic variance
(Section~\ref{main:dynamic-comparisons}). Retaining dated transforms
therefore implements the specified training protocol at a variance cost.

For real-time data, let $\tau_m$ be the first date on which all
components of a declared reference-period vector are available, with
$\tau_m<\tau_{m+1}$. If $L_{i,m}^{[v]}$ denotes the level for period
$m$ available at date $v$, lock each growth transformation as
\begin{equation}\label{rli:locked-growth}
 Y_{i,m}^{\rm L}=100\{\log L_{i,m}^{[\tau_m]}
                         -\log L_{i,m-1}^{[\tau_m]}\}.
\end{equation}
Both levels use the vintage available at $\tau_m$, so a growth rate
does not mix levels from different releases. Other declared
transformations complete the locked vector $Y_m^{\rm L}$. An early
component may have been revised before the final component appears.
The vector records the first common release rather than each component's
first estimate. Define the service history by
\begin{equation}\label{rli:service-field}
 \mathcal G_m=\sigma(Y_s^{\rm L}:s\le m).
\end{equation}
Proposition~\ref{rli:transfer} applies the theory to these records if
the locked sequence itself satisfies \eqref{main:var-model} and the
innovation conditions. The recording rule does not establish that law.
The target conditions on $\mathcal G_T$. Using all calendar-time news
would require the next innovation to retain its maintained law given
that larger information set. Replacing the locked history with revised
data changes the observed process and lies outside the comparison using the same
observations. The stationary argument uses a two-sided history, while
implementation starts with observed initial lags.

Any fixed $\eta\in(0,1)$ is admissible in the proofs. The baseline
training rule uses $\eta=0.7$. Online
Appendix~\ref{sec:sim-small-support} also reports $\eta=0.55$ and $0.85$
as sensitivity checks. The growing initial training sample ensures
uniform control of retained marginal fits, while $b_T/T\to0$ preserves
the first-order sample normalization.

\subsection{Fixed vine and sequential fit}

For continuous predictive margins $F_{i,t}$ with densities $f_{i,t}$,
a valid copula density defines
\begin{equation}\label{eq:joint-model}
 q_t(y\mid\mathcal F_{t-1})=
 c_{\theta,V,K}\{F_{1,t}(y_1),\ldots,F_{n,t}(y_n)\}
                       \prod_{i=1}^n f_{i,t}(y_i).
\end{equation}
The regular
vine $V$, depth $K\in\{0,\ldots,n-1\}$, and pair families are fixed
before estimation. The simplified vine factorization is
\begin{equation}\label{eq:vine-product}
 c_{\theta,V,K}(u)=\prod_{r=1}^K\prod_{e\in E_r(V)}
 c_{e,\theta_e}(u_{a_e|D_e},u_{b_e|D_e}),\qquad |D_e|=r-1.
\end{equation}
An edge links $a_e$ and $b_e$ conditional on $D_e$. The arguments
are obtained by successive conditional-CDF maps
$h_{a|b;D}(v,w)=\partial_w C_{ab;D}(v,w)$, following the pair-copula
construction of \citet{aas_pair-copula_2009}. The simplifying restriction
makes the pair parameter invariant to the realized conditioning
variables. It is a dependence assumption, not a consequence of removing
conditional means and variances \citep{czado_nagler_2022}.
\begin{ass}[Fixed and correctly specified vine]\label{main:vine-assumption}
The copula of $Z_t$ is \eqref{eq:vine-product} at $\theta_0$.
The vine $V$, depth $K$, and pair families are fixed, and omitted
edges are independence copulas. For the retained free parameters,
\[
 \Theta=\prod_{e=1}^E\Theta_e\subset\mathbb R^{q_\theta}
 \text{ is compact},\qquad
 \theta_0\in\operatorname{int}(\Theta),\qquad q_\theta<\infty.
\]
All these objects are fixed as $T\to\infty$.
Each free parameter belongs to one edge, and the product space permits
its variation independently of the other edge parameters. Pair densities
are positive on $(0,1)^2$ throughout their admissible parameter spaces.
\end{ass}
Correct specification requires the simplifying restriction and all
imposed independence edges to hold. The product parameter space allows
stagewise estimation of retained pairs, as in
\citet{hobaeck_haff_2013}. Shared parameters would require a different
estimating system. Procedures such as \citet{dissmann_selecting_2013}
select the vine structure and pair families from the data. Here, those
choices and the truncation depth are fixed, so the theory excludes
selection uncertainty. Positive pair densities make the log objective
well defined on the open unit square, without requiring a uniform
positive lower bound.

At origin $T$, enumerate retained edges in tree order and fit
\begin{equation}\label{eq:sequential-fit}
 \widehat\theta_e\in\argmax_{\theta_e\in\Theta_e}
 \frac1{T-b_T}\sum_{t=b_T+1}^T
 \log c_{e,\theta_e}
 \{\widehat U_{a_e|D_e,t}(\widehat\theta_{<e}),
   \widehat U_{b_e|D_e,t}(\widehat\theta_{<e})\}.
\end{equation}
Use fixed measurable tie-breaking. At edge $e$, optimization holds
the estimated ancestors $\widehat\theta_{<e}$ fixed. Their sampling
error still enters the score because they determine its conditional
inputs. The estimating system is therefore triangular and generally
differs from the score of a jointly maximized likelihood. Later origins
update estimates of the same constant dependence law.

The fitted orthant probability combines this dependence estimate with
the terminal margins:
\begin{equation}\label{main:event-estimator}
 \widehat p_{T+1}(c)=C_{\widehat\theta}\left(
       \Phi\left\{\frac{c_i-x_{T+1}'\widehat\beta_{i,T}}
                        {\sqrt{\widehat v_{i,T}}}\right\}:i\le n\right).
\end{equation}
Theorem~\ref{main:event-inference} gives its interval. Section~\ref{main:implementation}
states the numerical accuracy needed for approximate optimization and
probability evaluation. Time-varying coefficients and latent volatility
require a separate analysis of the marginal estimation errors.

\section{Inference for Conditional Joint Forecasts}\label{sec:theory}

Inference requires a joint law for the sequential vine fit and terminal
VAR estimates, followed by an expansion at the random forecast state.
We first derive the joint covariance, accounting for recursive marginal
error. We then establish the separation from the terminal state needed
for event intervals. The same decomposition yields the refit comparison,
finite training weights, and logit intervals.

\subsection{Scores and regularity}

Let $V_e(\theta_{<e},z)$ be the two conditional PITs at edge $e$,
starting from coordinates $\Phi(z_i)$. Define the edge objective and
its own-parameter score by
\begin{equation}\label{main:edge-scores}
 \ell_e(\theta_{\le e},z)=\log c_{e,\theta_e}\{V_e(\theta_{<e},z)\},
 \qquad g_e(\theta,z)=\partial_{\theta_e}\ell_e(\theta_{\le e},z).
\end{equation}
The stacked score $g=(g_1',\ldots,g_E')'$ has dimension
$q_\theta=\dim(\theta)$. Ancestor dependence makes its Jacobian lower
block triangular.

\begin{ass}[Sequential identification and curvature]\label{main:identification}
For each retained edge, put
$Q_e(\vartheta)=E\ell_e(\theta_{<e,0},\vartheta,Z_t)$.
\textup{(i)} The population objective identifies the edge parameter:
\[
 \argmax_{\vartheta\in\Theta_e}Q_e(\vartheta)=\{\theta_{e,0}\}.
\]
\textup{(ii)} Its local curvature is nonsingular:
\[
 I_e=-E\partial_{\theta_e}g_e(\theta_0,Z_t)\succ0.
\]
\end{ass}

\begin{ass}[Normal-score smoothness]\label{main:regularity}
Each $\ell_e$ is continuous on a neighborhood $\mathcal N$ of the
compact parameter product and has continuous mixed derivatives
satisfying, for finite $C,r$,
\[
 \sup_{\theta\in\mathcal N}
 \bigl\|\partial_\theta^\alpha\partial_z^\gamma
                  \ell_e(\theta_{\le e},z)\bigr\|
 \le C(1+\|z\|^r),\qquad
 |\alpha|\le3,\quad |\gamma|\le2,\quad |\alpha|+|\gamma|\le3.
\]
The constants may depend on the fixed model, but not on $T$ or $z$.
\end{ass}

Under the remaining regularity conditions,
Assumption~\ref{main:identification}(i) identifies the limit of each
stagewise maximizer. Part (ii) makes the local score system invertible
at that limit. Neither global concavity nor uniqueness of every
stationary point is required. The distinction between a population
maximizer and a score root also governs numerical implementation, as
discussed by \citet[Section 4.3]{nagler_kruger_min_2022}.

The smoothness condition covers the composite edge objective and every
ancestor transformation. An earlier pair estimate changes the inputs
to later pairs, so bounds for isolated pair densities would not suffice.
Normal-score coordinates make the full recursion tractable: derivatives
near a PIT boundary can be bounded by polynomials in $z$. Normal
coordinate margins give these polynomials all fixed moments, including
under non-Gaussian joint dependence.

A separate argument controls the recursive regressions. The stable VAR
filter and growing initial training sample give uniform bounds over the retained
fits in Lemmas~\ref{dyn:moments}--\ref{dyn:first-stage}. Finite-lag
approximation of that filter also establishes the joint CLT in
Appendix~\ref{dji:section}, without imposing an additional mixing-rate
condition. These bounds depend on the fixed stationary law. They do
not provide uniform coverage near a unit root or a pair-parameter
boundary. Correct copula specification is needed for a further reason:
a pseudo-true limit, permitted in parts of
\citet{nagler_kruger_min_2022}, need not deliver the conditional event
probability in \eqref{main:conditional-target}.

For $a=(a_\mu',a_v')'$, put
$[T_a(z)]_i=(z_i-a_{\mu,i})/(1+a_{v,i})^{1/2}$. The score moments are
\begin{align}\label{main:score-matrices}
 \varphi(z)&=(z_1,\ldots,z_n,z_1^2-1,\ldots,z_n^2-1)',\notag\\
 J&=E\partial_\theta g(\theta_0,Z_t),\qquad
 S=E\{g(\theta_0,Z_t)g(\theta_0,Z_t)'\},\notag\\
 \Sigma_\varphi&=E\{\varphi(Z_t)\varphi(Z_t)'\},\qquad
 \Gamma=E\{g(\theta_0,Z_t)\varphi(Z_t)'\},\notag\\
 B_g&=E\,\partial_a g\{\theta_0,T_a(Z_t)\}\big|_{a=0}.
\end{align}
The columns of $B_g$ measure score sensitivity to standardized location
and relative variance errors. The diagonal blocks of $J$ are the
own-edge curvatures $-I_e$, while its lower blocks transmit estimation
error from ancestors. Thus the correction has both a marginal stage
and a sequence of dependence stages, consistent with the distinction
between marginal and copula estimation emphasized by
\citet{joe_2005,patton_2012}. Separate pairwise standard errors omit
the latter propagation.

Proposition~\ref{score:gaussian-family} verifies these conditions for
Gaussian vines: conditional normal scores are linear in the original
scores, making the edge objectives quadratic. For Clayton vines,
Online Appendix~\ref{tail:section} uses logarithms of conditional PITs
to propagate polynomial bounds through every ancestor, including
nonindependent higher-tree pairs. Free Clayton parameters stay in
compact subsets of $(0,\infty)$. Independence edges can be fixed in
advance.
The Clayton pair is
\[
 C_\vartheta(u,v)=(u^{-\vartheta}+v^{-\vartheta}-1)^{-1/\vartheta},
 \qquad \vartheta>0;
\]
its density and conditional transforms are obtained by differentiation
on $(0,1)^2$. Only the CDF is extended continuously to the closed square.
The additional conditions for forecast probabilities are stated in
Section~\ref{main:forecast-section}.
Selecting independence at a parameter boundary is a different problem.

\subsection{The joint dynamic limit}\label{main:dynamic-section}

Write $Q=E(x_tx_t')$ and $e_1=(1,0,\ldots,0)'$. Define the terminal
errors and their influence vector by
\begin{align}\label{main:terminal-influences}
 d_{i,T}&=(\widehat\beta_{i,T}-\beta_i)/\sigma_i,\qquad
 a^v_{i,T}=\widehat v_{i,T}/\sigma_i^2-1,\notag\\
 \delta_T&=(d_{1,T}',\ldots,d_{n,T}',a^v_{1,T},\ldots,a^v_{n,T})',\notag\\
 v_t&=(Z_{1t}^2-1,\ldots,Z_{nt}^2-1)',\qquad
 w_{i,t}=Q^{-1}x_tZ_{it},\notag\\
 e_t&=(w_{1,t}',\ldots,w_{n,t}',v_t')',\qquad
 \Delta_T=((\widehat\theta-\theta_0)',\delta_T')'.
\end{align}
Coefficient blocks are stacked by response equation. The joint
covariance has the decomposition
\begin{align}
 \xi_t&=\begin{pmatrix}
 -J^{-1}\{g(\theta_0,Z_t)+B_g\varphi(Z_t)\}\\ e_t
 \end{pmatrix},\notag\\
 \mathcal W_{\rm VAR}
 &=E(\xi_t\xi_t')+
 \begin{pmatrix}J^{-1}B_g\Sigma_\varphi B_g'J^{-\prime}&0\\0&0\end{pmatrix}.
 \label{main:covariance-decomposition}
\end{align}
The outer-product term combines marginal estimation and the copula
score, retaining their cross covariances. The second term captures the
additional variation from dated fits and affects only the copula block.
The first term is the terminal-refit covariance for either maintained
family (Section~\ref{main:dynamic-comparisons}). The decomposition
provides a direct positive semidefinite construction.
To display its blocks, let
$R_Z=E(Z_tZ_t')$, $M_3=E(Z_tv_t')$, and $S_v=E(v_tv_t')$.
Partition $M=\Gamma+B_g\Sigma_\varphi=(M_\mu,M_v)$, with $n$ columns
in each part, and define
\begin{align}\label{main:joint-covariance}
 W_\delta&=\begin{pmatrix}
 R_Z\otimes Q^{-1}&M_3\otimes e_1\\
 M_3'\otimes e_1'&S_v
 \end{pmatrix},\notag\\
 D_{\rm VAR}&=-J^{-1}
       (M_{\mu,1}e_1',\ldots,M_{\mu,n}e_1',M_v),\notag\\
 \mathcal W_{\rm VAR}&=
 \begin{pmatrix}V_{\rm seq}&D_{\rm VAR}\\D_{\rm VAR}'&W_\delta\end{pmatrix},
\end{align}
where $M_{\mu,i}$ denotes column $i$ and
\begin{equation}\label{main:sequential-covariance}
 V_{\rm seq}=J^{-1}
 \{S+\Gamma B_g'+B_g\Gamma'+2B_g\Sigma_\varphi B_g'\}J^{-\prime}.
\end{equation}
The block $M_3$ links coefficient and scale uncertainty. Asymmetric
dependence can produce nonzero mixed third moments even with normal
coordinate margins.

\begin{thm}[Joint inference after recursive VAR estimation]\label{main:dynamic}\label{dji:joint}
Under Assumptions~\ref{main:var-assumption}--\ref{main:regularity},
with the estimators and $b_T$ defined in Section~\ref{sec:method}, the sequential estimator is
consistent. With $N_T=T-b_T$ and
$h_{T,s}=\sum_{m=\max(b_T,s)}^{T-1}m^{-1}$, interpreted as zero for an
empty sum,
\begin{align}\label{main:sequential-expansion}
 \sqrt T(\widehat\theta-\theta_0)
 &=-\frac{J^{-1}}{\sqrt T}\sum_{s=1}^T
 \{1\{s>b_T\}g(\theta_0,Z_s)+h_{T,s}B_g\varphi(Z_s)\}+o_p(1),
 \\
 \sqrt T\delta_T&=T^{-1/2}\sum_{s=1}^Te_s+o_p(1),\qquad
 \sqrt T\Delta_T\Longrightarrow N(0,\mathcal W_{\rm VAR}).
 \label{main:joint-limit}
\end{align}
The covariance estimator in \eqref{main:joint-feasible} is positive
semidefinite and consistent. The terminal fitted conditional margins
satisfy
\begin{equation}\label{main:dynamic-margin-rate}
 \max_{i\le n}\sup_y
 |\widehat F_{i,T+1|T}(y)-F^0_{i,T+1}(y)|=O_p(T^{-1/2}).
\end{equation}
\end{thm}

An innovation affects the dependence estimate both through its own
copula score and through the marginal fits used to standardize later
observations. Its contribution to a fit based on the first $m$
observations is proportional to $1/m$. Summing over the retained fits
that reuse innovation $s$ gives the harmonic weight $h_{T,s}$ in
\eqref{main:sequential-expansion}. Earlier innovations enter more of
these fits and receive greater weight. The normalized sum of squared
weights tends to two, whereas the corresponding average over retained
score dates tends to one. These limits explain the coefficients in
\eqref{main:sequential-covariance}. The recursive weighting mechanism
is familiar from \citet{west_1996}. Here it applies to the marginal
innovation moments and their joint covariance with the terminal VAR fit.

The reduction of the historical regression error to $\varphi(Z_s)$
depends on averaging over the lag designs. The intercept identity
$Qe_1=E(x_t)$ implies $E(x_t)'Q^{-1}x_s=1$, so the mean design leaves
only an innovation location term. Stability makes the normalized
contribution of the centered designs negligible, as established in
Lemma~\ref{dyn:interaction}. Location and scale errors therefore remain
in the general historical correction, even though every recursive regression
estimates the full coefficient vector. Gaussian sign symmetry makes its
location columns zero. The remaining scale contribution is explicit in
Proposition~\ref{rli:gaussian-primitives} for the bivariate case.

The terminal forecast is evaluated at one observed lag state, where
there is no such averaging. Errors in both the intercept and the slopes
affect its conditional mean, and their covariance retains the full
block $R_Z\otimes Q^{-1}$. The historical and terminal fits also share
innovations, which accounts for the cross block in
\eqref{main:joint-covariance}. The covariance comparisons in
Section~\ref{sec:sim-components} distinguish these contributions at a
common point forecast. Appendix~\ref{dji:section} gives the complete
argument, with bounds for recursive estimation and moment consistency in Online
Appendix~\ref{supp:joint-proofs}.

For implementation, use the stored normal scores, mathematically
$\widehat Z_{it}=\Phi^{-1}(\widehat U_{it})$, and put
$\widehat g_t=g(\widehat\theta,\widehat Z_t)$ and
$\widehat\varphi_t=\varphi(\widehat Z_t)$. Estimate $J$, $B_g$, and
$\Sigma_\varphi$ by the retained-sample averages of their defining
derivatives and outer products in \eqref{main:score-matrices}.
Let $\widehat e_t$ substitute $\widehat Q_T^{-1}$ and $\widehat Z_t$
in \eqref{main:terminal-influences}. A feasible joint covariance is
\begin{align}\label{main:joint-feasible}
 \widehat\xi_t&=
 \begin{pmatrix}-\widehat J^{-1}
       (\widehat g_t+\widehat B_g\widehat\varphi_t)\\\widehat e_t\end{pmatrix},
 \notag\\
 \widehat{\mathcal W}_{\rm VAR}
 &=\frac1{N_T}\sum_{t>b_T}\widehat\xi_t\widehat\xi_t'
 +\begin{pmatrix}
 \widehat J^{-1}\widehat B_g\widehat\Sigma_\varphi
                  \widehat B_g'\widehat J^{-\prime}&0\\0&0
 \end{pmatrix}.
\end{align}
Use the identity matrix if a required inverse is absent. The probability
of this fallback tends to zero. Each term in
\eqref{main:joint-feasible} is positive semidefinite, so the construction
avoids a separate adjustment to enforce that property. The common outer
product already contains the copula--terminal cross covariances.
Only the copula block receives the additional historical correction.

\subsection{A confidence interval at the random forecast state}
\label{main:forecast-section}

For a lag state $r\in\mathbb R^{d_x}$, write
$z_i(r)=(c_i-r'\beta_i)/\sigma_i$ and $u_i(r)=\Phi\{z_i(r)\}$.
Let $C_i(\theta,u)=\partial C_\theta(u)/\partial u_i$.
The gradient in the normalized coordinates of $\Delta_T$ has blocks
\begin{align}\label{main:event-gradient}
 d_\theta(r)&=\partial_\theta C_{\theta_0}\{u(r)\},\notag\\
 d_{\beta_i}(r)&=-C_i(\theta_0,u(r))\phi\{z_i(r)\}\,r,\notag\\
 d_{v_i}(r)&=-\tfrac12 C_i(\theta_0,u(r))z_i(r)\phi\{z_i(r)\}.
\end{align}
Stack these columns in the same order as $\Delta_T$ and put
$\tau^2(r)=d(r)'\mathcal W_{\rm VAR}d(r)$.
The coefficient gradient includes the entire state $r$.
Let $X$ have the stationary law of $x_t$.

\begin{ass}[Event smoothness and nondegeneracy]\label{main:event-regularity}
\textup{(i)} The map $(\theta,u)\mapsto C_\theta(u)$ is twice
continuously differentiable on $\mathcal N_0\times(0,1)^n$, for an
open neighborhood $\mathcal N_0$ of $\theta_0$.
\textup{(ii)} For the fixed event under consideration,
$P\{\tau^2(X)>0\}=1$.
\end{ass}
Part (i) permits a Taylor expansion of the event probability. Part (ii)
is needed for studentization and concerns the variance of this event
alone, so the full joint covariance may be singular. No positive lower
bound over all states is required. As in the prediction results of
\citet{nagler_kruger_min_2022}, estimation regularity and smoothness of
the prediction functional are separate requirements. The remaining
issue here is the joint behavior of estimation error and the terminal
state, resolved by Lemma~\ref{dfe:state-separation}.

\begin{thm}[Conditional event inference]\label{main:event-inference}\label{dfe:event-clt}
Under Theorem~\ref{main:dynamic} and
Assumption~\ref{main:event-regularity}(i), let
$G\sim N(0,\mathcal W_{\rm VAR})$ be independent of $X$. Then
\begin{equation}\label{main:event-limit}
 \bigl(\sqrt T\{\widehat p_{T+1}(c)-p_{T+1}(c)\},x_{T+1}\bigr)
 \Longrightarrow (d(X)'G,X).
\end{equation}
Let $\widehat d_T$ substitute the fitted parameters and terminal
standardized thresholds in \eqref{main:event-gradient}, and define
$\widehat\tau_T^2=\widehat d_T(x_{T+1})'
\widehat{\mathcal W}_{\rm VAR}\widehat d_T(x_{T+1})$.
Then $\widehat\tau_T^2-\tau^2(x_{T+1})\to_p0$.
If Assumption~\ref{main:event-regularity}(ii) also holds, then for every
fixed $0<\alpha<1$,
\begin{align}
 \frac{\sqrt T\{\widehat p_{T+1}(c)-p_{T+1}(c)\}}{\widehat\tau_T}
 &\Longrightarrow N(0,1),\label{main:pivot}\\
 P\left\{p_{T+1}(c)\in
 [\widehat p_{T+1}(c)\pm z_{1-\alpha/2}\widehat\tau_T/\sqrt T]
                   \cap[0,1]\right\}&\longrightarrow1-\alpha.
 \label{main:interval}
\end{align}
Set the statistic to zero when its estimated variance is zero. That
event has vanishing probability under the positivity condition.
\end{thm}

The estimates and terminal state share observations, yet their leading
limits depend on separate parts of the innovation history. Estimation
error accumulates over the full sample, so deleting a slowly growing
block of recent innovations changes its normalized influence negligibly.
Stability allows the terminal state to be approximated by those same
recent innovations, since the distant past has a geometrically decaying
effect. The two approximations are independent, and both approximation
errors vanish. This gives the independence of $G$ and $X$ in
\eqref{main:event-limit} without requiring finite-sample independence.

For a given limiting state $X$, the leading forecast error is the
Gaussian projection $d(X)'G$, with variance $\tau^2(X)$. Different
states change the sensitivity of the forecast to parameter error,
producing a mixed normal distribution across samples. The feasible standard
error adapts to the realized state, so studentization removes this
random scale and yields \eqref{main:pivot}. Almost-sure positivity of
$\tau^2(X)$ suffices. States with arbitrarily small positive variance
can be handled by localization, without imposing a uniform lower bound.
Appendix~\ref{dfe:section} establishes the separation and studentization
arguments. The resulting coverage statement concerns a conditional
probability evaluated at a random realized state. It does not guarantee
coverage conditional on every training history or provide a prediction
set for the next binary outcome. Section~\ref{sec:sim-coverage} examines
this repeated-sample coverage in finite samples.

For the family verifications and rectangle extension, a sufficient
smoothness condition on the full log density
$\ell^*(\theta,z)=\sum_e\ell_e(\theta_{\le e},z)$ is
\begin{equation}\label{main:event-envelope}
 \sup_{\theta\in\mathcal N_0}
 \|\partial_\theta^\alpha\partial_z^\gamma\ell^*(\theta,z)\|
 \le C(1+\|z\|^r),\qquad |\alpha|+|\gamma|\le3,
\end{equation}
with continuous derivatives on a neighborhood of a compact parameter
set containing $\theta_0$ in its interior. This includes pure third
derivatives in $z$, which Assumption~\ref{main:regularity} does not
require. Lemma~\ref{pred:cdf-smooth} shows that it implies
Assumption~\ref{main:event-regularity}(i). Gaussian log densities are
quadratic in $z$. Lemma~\ref{tail:propagation} verifies the bound for
the positive Clayton vines, including the ancestor transformations.

For a fixed rectangle $A=\prod_i(\ell_i,u_i]$, replace the CDF by
its probability over $A$ and use the corresponding gradient.
Corollary~\ref{dfe:rectangle-clt} proves the same interval, treating
infinite endpoints as inactive boundaries, under
\eqref{main:event-envelope} and the event-specific condition
$P\{\tau_A^2(X)>0\}=1$. This covers fixed mixed directions and target subsets without
deleting factors from a larger fitted vine. Thresholds estimated from
the same observations or moving into the tail require additional
influence or uniformity arguments.

\subsection{Dynamic refitting and finite training weights}\label{main:dynamic-comparisons}

A terminal refit provides a feasible benchmark that keeps the copula
family and the population target fixed. Keep the observations and
terminal marginal estimates fixed, recompute each standardized residual,
and fit the same sequential vine:
\begin{equation}
 \widetilde Z_{it}=\frac{y_{it}-x_t'\widehat\beta_{i,T}}
                            {\sqrt{\widehat v_{i,T}}},\quad 1\le t\le T,
 \qquad
 \widetilde\theta_e\in\arg\max_{\theta_e\in\Theta_e}
 \frac1T\sum_{t=1}^T
       \ell_e(\widetilde\theta_{<e},\theta_e,\widetilde Z_t).
 \label{dfs:estimator}
\end{equation}
All $T$ residuals enter this comparator, with the same parameter spaces
and measurable tie-breaking as in \eqref{eq:sequential-fit}.
Proposition~\ref{dfs:theorem} proves the Gaussian result, and Online
Proposition~\ref{refit:theorem} establishes the same expansion for
positive Clayton vines and for a refit using only the archive's retained
dates. The joint influence vector has copula block $-J^{-1}\{g_t+B_g\varphi_t\}$ and terminal
block $e_t$, together with the covariance ordering
\begin{equation}\label{main:dynamic-archive-cost}
 \mathcal W_{\rm VAR}-\mathcal W_{\rm fs}
 =\begin{pmatrix}J^{-1}B_g\Sigma_\varphi B_g'J^{-\prime}&0\\0&0\end{pmatrix}
 \succeq0.
\end{equation}
Evaluate $\widehat g_t,\widehat\varphi_t,\widehat J,\widehat B_g$
at $(\widetilde\theta,\widetilde Z_t)$ using all $T$ observations,
and form $\widehat e_t$ from the same residuals and $\widehat Q_T^{-1}$.
The feasible full-sample covariance is
\begin{equation}
 \widehat\xi_t^{\rm fs}=
 \begin{pmatrix}-\widehat J^{-1}
       (\widehat g_t+\widehat B_g\widehat\varphi_t)\\\widehat e_t\end{pmatrix},
 \qquad
 \widehat{\mathcal W}_{\rm fs}
       =T^{-1}\sum_{t=1}^T\widehat\xi_t^{\rm fs}
                              \widehat\xi_t^{{\rm fs}\prime}.
 \label{dfs:feasible}
\end{equation}
Use the identity-matrix convention when a required inverse is absent.
The refit covariance retains every terminal slope and scale. Terminal
marginal errors are common across residuals. Averaging their score
derivatives and using the intercept identity gives the correction
$B_g\varphi_t$, with unit weights rather than harmonic weights.
Equation~\eqref{main:dynamic-archive-cost}
then compares limiting covariances at the same population parameters.
The extra archive term measures the variance cost of training the copula
on successive out-of-sample marginal forecasts while updating the
terminal margins.
The ordering need not extend to finite-sample standard errors evaluated at separately
fitted parameters. Nor does it compare estimators trained on different
data vintages or forecasting laws.

To isolate the transformation rule in a finite sample, also fit the
terminal residuals on the same $N=T-b_T$ dates as the archive. Marginal
estimation still uses all $T$ observations. The covariance of its leading
array is
\begin{equation}\label{main:matched-refit-covariance}
 \mathcal W_{{\rm rf},T}^{(b_T)}
 =\mathcal W_{\rm fs}+\left(\frac TN-1\right)
 \begin{pmatrix}J^{-1}SJ^{-\prime}&0\\0&0\end{pmatrix}.
\end{equation}
The additional score term reflects the shorter copula sample. Online
Section~\ref{supp:matched-refit} gives the influence expansion, a positive
semidefinite feasible estimator, and the event-inference proof.
The matched-date and all-observation refits have the same limiting
covariance because $b_T/T\to0$. The simulations report both to distinguish
the use of extra dates from the use of revised transforms.

The initial training period may occupy a substantial fraction of a
practical sample, even though $b_T/T\to0$. Retaining that fraction in
the leading array gives a finite-weight covariance. Write
$b=b_T$, $N=T-b$, and $H_T=\sum_{m=b}^{T-1}m^{-1}$.
Keeping the exact $T/N$ normalization replaces the copula-score
covariance by
\begin{align}\label{main:finite-weights}
 \Omega_T^{\rm fin}
 &=k_{11}S+k_{12}(\Gamma B_g'+B_g\Gamma')+
                         k_{22}B_g\Sigma_\varphi B_g',\notag\\
 k_{11}&=T/N,\qquad k_{12}=T(N-bH_T)/N^2,\qquad
 k_{22}=T\{2N-(2b-1)H_T\}/N^2.
\end{align}
The copula--terminal cross coefficients remain one, and the terminal
block is unchanged. To use one common empirical moment matrix, set
$a_{T,s}=(T/N)\mathbf1\{s>b\}$ and $d_{T,s}=(T/N)h_{T,s}$ and form
\begin{align}
 \widehat\zeta_t&=(\widehat g_t',\widehat\varphi_t',\widehat e_t')',&
 \widehat C_T&=\frac1N\sum_{t>b}\widehat\zeta_t\widehat\zeta_t',\notag\\
 \widehat L_{T,s}
 &=\begin{pmatrix}
 -a_{T,s}\widehat J^{-1}&
 -d_{T,s}\widehat J^{-1}\widehat B_g&0\\
 0&0&I_{n d_x+n}
 \end{pmatrix},&
 \widehat{\mathcal W}_T^{\rm fin}
 &=\frac1T\sum_{s=1}^T
       \widehat L_{T,s}\widehat C_T\widehat L_{T,s}'.
 \label{fdw:feasible}
\end{align}
The score, nuisance and terminal quantities use the same retained
observations as \eqref{main:joint-feasible}. Use its identity-matrix convention.
In computation, partition $\widehat C_T$ into score, nuisance and
terminal blocks and apply the coefficients in \eqref{main:finite-weights},
without evaluating the displayed sum date by date.
Corollary~\ref{fdw:feasibility} proves that this estimator is positive
semidefinite and consistent and can replace
\eqref{main:joint-feasible} in the event interval. At $T=1000$, the
training rule $b_T=\lceil T^{0.7}\rceil$ gives $b=126$ and coefficients
$1.1442$, $0.8019$, and $1.6065$, instead of $1,1,2$.
These coefficients are exact for the leading array. Since the
linearization remainder is omitted, they provide neither an exact
finite-sample variance nor a coverage improvement theorem.
Section~\ref{sec:sim-small} compares the resulting intervals with
limiting-weight and terminal-refit intervals.

The terminal slopes also explain why one-sided event intervals are
nondegenerate. Centering the lag design isolates a coefficient influence
orthogonal to every function of the current innovation, including the
copula scores and scale influences. Its projection on the event gradient
is positive except at the mean lag state, which has probability zero.
Proposition~\ref{ond:proposition} therefore verifies the required
almost-sure positivity for Gaussian and positive Clayton archives.
Gaussian sign symmetry gives the stronger conclusion of positivity at
every finite state for both archive and refit
(Proposition~\ref{gnd:proposition}). General bounded rectangles retain
the event-specific condition.

\subsection{Logit intervals and the variance cost of the archive}
\label{main:archive-interpretation}

Near zero or one, a symmetric probability-scale interval can have a
different finite-sample behavior from an interval constructed after
transformation. Applying the normal approximation to log odds gives
asymmetric probability endpoints and preserves the first-order coverage
result. Write $\widehat p_T$ for the fitted event probability
at $x_{T+1}$, and put
$\ell(p)=\log\{p/(1-p)\}$, $\Lambda(u)=(1+e^{-u})^{-1}$, and
$\widehat s_T=\widehat\tau_T/\sqrt T$. For an interior fitted probability,
the logit interval is
\begin{equation}\label{main:logit-interval}
 I_{\ell,T}=\left[
 \Lambda\left\{\ell(\widehat p_T)-
       \frac{z_{1-\alpha/2}\widehat s_T}{\widehat p_T(1-\widehat p_T)}\right\},\quad
 \Lambda\left\{\ell(\widehat p_T)+
       \frac{z_{1-\alpha/2}\widehat s_T}{\widehat p_T(1-\widehat p_T)}\right\}
 \right].
\end{equation}
Proposition~\ref{rli:logit} proves its repeated-sample coverage under the
corresponding event theorem, $0<p_A(X)<1$ and $\tau_A^2(X)>0$ almost
surely. Its proof localizes away from probability boundaries and small
projected variances, so no global lower bound over all lag states is
imposed. A deterministic clipping threshold tending to zero defines the formula at numerical endpoints. Section~\ref{sec:sim-logit} examines finite-sample coverage.

The size of the archive correction depends on the event as well as the
forecast state. Partition the event gradient
into dependence and marginal components, and set
\[
 q_A(r)=\|\Sigma_\varphi^{1/2}B_g'J^{-\prime}d_{A,\theta}(r)\|^2.
\]
Proposition~\ref{rli:cost} gives
$\tau_{\rm ar,A}^2(r)=\tau_{\rm fs,A}^2(r)+q_A(r)$ for either maintained
family in the terminal-refit comparison. Thus $q_A(r)/\tau_{\rm ar,A}^2(r)$ measures the
share of event variance due specifically to preserving issued margins.
Corollary~\ref{rli:cost-coverage} quantifies limiting undercoverage when a refit
variance is used at the archived probability estimate. Holding that
estimate fixed also permits an exact calculation of when omitting the
archive term changes exclusion of a reference probability. These
comparisons isolate a covariance omission from changes in the fitted
forecast. The application reports both the archive share and the effect
of removing terminal-margin uncertainty.

The bivariate Gaussian case explains what makes this share small or
large. Let $\rho$ denote the estimated pair's true correlation,
$z_i(r)=(c_i-r'\beta_i)/\sigma_i$, and
$h(r)=\phi_2\{z_1(r),z_2(r);\rho\}$, where $\phi_2$ is the
standard bivariate normal density. Proposition~\ref{rli:gaussian-primitives}
gives
\begin{equation}\label{main:bivariate-archive-cost}
 q_A(r)=h(r)^2\frac{\rho^2(1-\rho^2)^2}{1+\rho^2}.
\end{equation}
The second factor measures the additional correlation-estimation
variance caused by dated scale estimates. Historical location errors
make no first-order contribution in the Gaussian model. The first
factor measures the event's squared sensitivity to correlation.
The first-order archive cost is therefore zero at independence, even
when correlation is freely estimated. Away from independence, weak
event sensitivity can still make the cost small.

The archive share also depends on the variance common to both estimators.
To isolate this denominator, consider the stationary submodel $y_t=Z_t$
while estimating every coefficient and scale of the full bivariate
VAR($p$). For the fixed event $y_{1,T+1}\le0$, $y_{2,T+1}\le0$,
the true standardized thresholds are zero at every lag state.

With state leverage $\lambda_x(r)=r'Q^{-1}r\ge1$, the variances satisfy
\begin{equation}
\begin{aligned}
 q_A(r)&=\frac{\rho^2(1-\rho^2)}{4\pi^2(1+\rho^2)},\qquad
 \tau_{{\rm fs},A}^2(r)=\frac{1-\rho^2}{4\pi^2}
                         +\frac{\lambda_x(r)(1+\rho)}{4\pi},\\
 \frac{q_A(r)}{\tau_{{\rm ar},A}^2(r)}
 &=\frac{\rho^2/(1+\rho^2)}
 {1+\rho^2/(1+\rho^2)+\pi\lambda_x(r)/(1-\rho)}.
\end{aligned}
 \label{main:bivariate-archive-share}
\end{equation}

Increasing leverage raises the common terminal-coefficient variance
without changing $q_A(r)$. With $\rho=0.5$, the archive share falls from
$2.67\%$ at $\lambda_x=1$ to $1.00\%$ at $\lambda_x=3$. At
$\rho=-0.5$ and the same two leverage values, it is $6.07\%$ and
$2.67\%$, despite the same added variance. For nonnegative correlation,
Corollary~\ref{rli:gaussian-zero} bounds the share by
$4/(27\pi\lambda_x+4)$ in this submodel. A modest share can therefore
reflect substantial common marginal uncertainty, even when the archive
correction is nonzero. These population calculations explain a mechanism
within the maintained class. The three-variable simulation and
application magnitudes require their own event gradients and covariance
blocks.

\subsection{Numerical implementation}\label{main:implementation}

Implementation follows the same sequence as the theory. Fit the vine
to the dated transforms, then use their stored normal scores and lag
rows to form \eqref{main:joint-feasible}. Evaluate the event probability
and its gradient at $x_{T+1}$ using the terminal marginal estimates.
The resulting standard error gives either
\eqref{main:interval} or \eqref{main:logit-interval}.
\citet{stoeber_schepsmeier_2013} give algorithms for analytic scores
and observed information in regular-vine models. Here, the covariance
calculation also requires the derivatives with respect to the dated
marginal inputs and the terminal forecast parameters. Analytical and
numerical derivatives are both admissible under the accuracy conditions
below.

Propositions~\ref{op:numerical-score}--\ref{op:numerical-event} justify
a consistent local numerical solution whose average score residual is
$o_p(T^{-1/2})$, probability error is $o_p(T^{-1/2})$, and gradient
and covariance errors are $o_p(1)$. A small residual alone does not
select the consistent local solution. Consistency or a vanishing
stagewise objective gap is also required. The bound in \eqref{op:finite-difference-bound} also justifies a forward
difference with step $s_T=T^{-1/4}$ when probability-evaluation error
at the required points is $o_p(T^{-1/2})$.
For event and gradient evaluation, these bounds hold uniformly on
each compact set of lag states and a neighborhood of the true
parameters. Stationarity then permits evaluation at the random
terminal state.

\section{Monte Carlo Simulations}\label{sec:simulation}

The simulations ask how retaining dated transforms affects event
probability accuracy and interval coverage. Same-family refits isolate
that timing choice. Comparisons between Gaussian and Clayton pairs
under a fixed timing rule isolate dependence specification. Stationary
vines provide a published alternative for modeling serial and
cross-variable dependence. We then examine which estimation errors
matter for inference and how the procedure performs in larger systems.
The results show a small accuracy cost of retaining dated transforms,
large downside gains from the correct copula family, and substantial
coverage gains from accounting for current marginal estimation.

\subsection{Design and comparisons}\label{sec:sim-design}\label{replication:sim_dgp}

The baseline data-generating process is
\begin{equation}
 y_t=r y_{t-1}+\sqrt{1-r^2}\,Z_t,\qquad
 Z_{it}=\Phi^{-1}(U_{it}),\qquad U_t\overset{\mathrm{iid}}{\sim}C_{\theta_0},
 \label{sim:approved-dgp}
\end{equation}
with three variables, $r\in\{0.25,0.75\}$, and
$T\in\{250,1000,4000\}$. The stationary variance of each variable is
one. The innovation copula is a D-vine with pairs $(1,2)$, $(2,3)$,
and $(1,3\mid2)$. All pairs are Gaussian or positive Clayton, with
Kendall's tau $(0.30,0.30,0.20)$. The corresponding parameters are
$\sin(\pi\tau_e/2)$ for Gaussian pairs and $2\tau_e/(1-\tau_e)$
for Clayton pairs. Estimating the second-tree pair allows estimation
error in the first-tree pairs to affect subsequent estimates.

Each design uses 4,000 independent Monte Carlo replications after
discarding 1,000 initial observations generated from zero. Estimation
fits an intercept and all lag coefficients in each VAR equation,
including coefficients that are zero in the data-generating process.
The initial training sample has $b_T=\lceil T^{0.7}\rceil$ observations.
The proposed estimator, labeled Archive in the tables, retains each
historical transform at the marginal estimates available when it was
formed. It uses the terminal marginal estimates to forecast the next
observation.

Zero denotes the event $y_{i,T+1}\leq0$ for all $i=1,2,3$, and
Downside denotes $y_{i,T+1}\leq-1$ for all three variables. The
supplement also reports Mixed, the event $y_{1,T+1}\leq0$ and
$y_{3,T+1}>0$, with variable 2 unrestricted. In each replication, the
true probability is computed from \eqref{sim:approved-dgp} at the
observed terminal state $y_T$. Root mean squared error (RMSE) measures
the difference between the estimated and true probabilities. Smaller
RMSE indicates greater forecast accuracy. Coverage is the proportion
of replications whose confidence interval contains the true probability,
with unavailable intervals counted as noncoverage. All intervals have
nominal coverage 95\%, so coverage closer to 95\% is better. Mean
length uses available intervals and measures precision. Monte Carlo
standard errors measure simulation uncertainty, with paired calculations
when methods use the same simulated samples.

Three interval variants hold the proposed point estimate fixed. Wald
constructs symmetric endpoints on the probability scale and intersects
the interval with $[0,1]$. Logit constructs endpoints on the log-odds
scale and transforms them back to probabilities, as in
Proposition~\ref{rli:logit}. Both use the limiting-weight covariance.
Finite logit instead uses the actual training weights in
\eqref{fdw:feasible}. These comparisons assess the effects of the
interval transformation and the treatment of the initial training sample.

The same-family comparisons use the first 2,000 replications of every
design. Matched refit recalculates the residuals at the terminal marginal
estimates but keeps the same $T-b_T$ copula training dates. Full refit
uses all $T$ residuals. Both fit the same Gaussian or Clayton family as
Archive, and both estimate margins on all $T$ observations. Their
covariance follows Online Proposition~\ref{refit:theorem}. Comparing
Archive with Matched refit isolates transform timing. Comparing the
two refits isolates the additional dates. In Clayton designs we also
fit Gaussian pairs under each of the three timing rules, holding the
remaining specification fixed. This separates copula-family effects
from the effects of retaining dated transforms.

The known-margin oracle receives the true VAR coefficients and
innovation scales but estimates the copula on the same $T-b_T$ dates.
Its supplementary comparison uses the original 4,000 replications and
measures the cost of estimating the marginal model.

The stationary-vine benchmark follows \citet{nagler_kruger_min_2022}.
It models temporal and cross-variable dependence directly in observed
levels. We use order one, normal unconditional margins, Gaussian pairs,
and variable order $(1,2,3)$, which represents the Gaussian VAR law.
Each Gaussian design uses the same 2,000 replications as the refit
comparisons. Every scalar pair likelihood is maximized over its full
admissible correlation interval. Two percentile-interval constructions
start from that same point fit and use 999 draws. S-vine J uses the full
Jacobian of the sequential score equations and Gaussian multipliers.
S-vine N uses the native \texttt{svines} bootstrap, including its Hessian
and multiplier rule. Both retain valid original draws without replacement.
Online Section~\ref{supp:sim-benchmark} gives the construction and
availability rule. Comparing both algorithms assesses the sensitivity
of inference for the stationary-vine forecast.

\subsection{Coverage and precision}
\label{sec:sim-coverage}\label{sec:sim-small}\label{sec:sim-logit}

Logit intervals improve downside coverage relative to Wald intervals
(Table~\ref{tab:simulation-coverage}). Across the baseline designs,
logit coverage for Zero and Downside ranges from 94.1\% to 95.5\%,
and interval lengths decrease as $T$ increases. The largest improvement
occurs for Downside at $T=250$ and $r=0.75$. Gaussian coverage rises
from 84.8\% to 94.2\%, and Clayton coverage rises from 85.6\% to
94.2\%, with small increases in mean length. Thus the logit
transformation gives more accurate coverage for this tail event while
using the same point estimate and covariance. These findings support
the interval approximation in Theorem~\ref{main:event-inference} and
Proposition~\ref{rli:logit}.

Finite training weights change baseline coverage by at most 0.3
percentage points and leave interval lengths nearly unchanged. They
therefore provide similar performance to the limiting weights here,
consistent with their first-order equivalence in
Corollary~\ref{fdw:feasibility}. The supplement examines other training
lengths and the Mixed event, where logit intervals do not consistently
improve on Wald intervals.

\begin{table}[!htbp]
\centering
\caption{Coverage and length of archive intervals}\label{tab:simulation-coverage}\label{tab:small-sample-coverage}
\small
\setlength{\tabcolsep}{4pt}
\renewcommand{\arraystretch}{1.06}
\begin{tabularx}{\linewidth}{@{\extracolsep{\fill}}lrrrrr}
\toprule
Copula & $r$ & $T$ & Wald & Logit & Finite logit\\
\midrule
\midrule
\multicolumn{6}{@{\extracolsep{\fill}}l}{\textit{Panel A: Zero event}}\\
Gaussian & 0.25 & 250 & 93.8 [0.145] & 95.3 [0.144] & 95.3 [0.145]\\
Gaussian & 0.25 & 1000 & 94.5 [0.071] & 94.8 [0.071] & 94.8 [0.071]\\
Gaussian & 0.25 & 4000 & 95.2 [0.035] & 95.4 [0.035] & 95.3 [0.035]\\
Gaussian & 0.75 & 250 & 91.6 [0.099] & 94.7 [0.100] & 94.7 [0.100]\\
Gaussian & 0.75 & 1000 & 94.0 [0.049] & 95.1 [0.049] & 95.1 [0.049]\\
Gaussian & 0.75 & 4000 & 94.7 [0.024] & 95.2 [0.024] & 95.2 [0.024]\\
Clayton & 0.25 & 250 & 93.6 [0.136] & 95.3 [0.136] & 95.3 [0.136]\\
Clayton & 0.25 & 1000 & 94.1 [0.065] & 94.6 [0.065] & 94.5 [0.065]\\
Clayton & 0.25 & 4000 & 94.9 [0.033] & 95.2 [0.033] & 95.1 [0.032]\\
Clayton & 0.75 & 250 & 91.6 [0.098] & 94.6 [0.099] & 94.5 [0.099]\\
Clayton & 0.75 & 1000 & 94.5 [0.048] & 94.8 [0.048] & 94.7 [0.048]\\
Clayton & 0.75 & 4000 & 94.5 [0.024] & 94.7 [0.024] & 94.6 [0.024]\\
\midrule
\multicolumn{6}{@{\extracolsep{\fill}}l}{\textit{Panel B: Downside event}}\\
Gaussian & 0.25 & 250 & 88.8 [0.036] & 95.0 [0.039] & 95.1 [0.039]\\
Gaussian & 0.25 & 1000 & 92.7 [0.018] & 94.8 [0.018] & 94.8 [0.018]\\
Gaussian & 0.25 & 4000 & 94.7 [0.009] & 95.3 [0.009] & 95.3 [0.009]\\
Gaussian & 0.75 & 250 & 84.8 [0.027] & 94.2 [0.028] & 94.4 [0.028]\\
Gaussian & 0.75 & 1000 & 91.2 [0.013] & 94.9 [0.013] & 94.9 [0.013]\\
Gaussian & 0.75 & 4000 & 94.2 [0.006] & 95.5 [0.006] & 95.5 [0.006]\\
Clayton & 0.25 & 250 & 89.4 [0.051] & 94.6 [0.053] & 94.6 [0.053]\\
Clayton & 0.25 & 1000 & 92.9 [0.025] & 94.5 [0.025] & 94.4 [0.025]\\
Clayton & 0.25 & 4000 & 94.7 [0.013] & 95.3 [0.013] & 95.2 [0.013]\\
Clayton & 0.75 & 250 & 85.6 [0.033] & 94.2 [0.034] & 94.2 [0.034]\\
Clayton & 0.75 & 1000 & 92.7 [0.015] & 94.1 [0.015] & 94.1 [0.015]\\
Clayton & 0.75 & 4000 & 94.5 [0.008] & 95.0 [0.008] & 94.9 [0.008]\\
\bottomrule
\end{tabularx}
\par\smallskip\begin{minipage}{\linewidth}\footnotesize
\textit{Note:} Coverage is in percent for nominal 95\% intervals, with mean length in brackets. Coverage closer to 95\% is better, and shorter intervals indicate greater precision at comparable coverage. Unavailable intervals count as noncoverage, while mean length uses available intervals. Wald and Logit use the proposed point estimate and limiting-weight covariance, with endpoints constructed on the probability and log-odds scales, respectively. Finite logit uses log-odds endpoints and the finite-training-weight covariance in \eqref{fdw:feasible}. These variants compare interval construction while holding the point forecast fixed. Each row uses 4,000 common Monte Carlo replications. Zero and Downside have thresholds zero and $-1$ for all three variables. The largest coverage Monte Carlo standard error is 0.57 percentage points.
\end{minipage}
\end{table}

\subsection{Benchmark comparisons}\label{sec:sim-accuracy}

Same-family refitting generally improves probability accuracy by a
small amount (Table~\ref{tab:sim-benchmark-rmse}). Across both families,
the six sample-size and persistence settings, and the three events,
Matched refit has lower RMSE than Archive in 33 of 36 comparisons.
Full refit has lower RMSE in 35. Both have shorter mean intervals
in every comparison, while coverage rankings vary. Thus the proposed
method does not outperform same-family refitting on accuracy or
interval precision. This pattern is consistent with the additional
limiting covariance in \eqref{main:dynamic-archive-cost}.

Holding the copula dates fixed is informative because it removes the
refit's advantage from using more observations. The matched and full
refits differ in RMSE by at most 0.00023 over the entire grid.
For Clayton Downside at $r=0.25$, $T=1000$, RMSE is 0.00688
for Archive, 0.00670 for Matched refit, and 0.00663 for Full refit.
The Full-refit minus Archive difference is $-0.00025$, with paired
Monte Carlo standard error 0.00005. Most of that improvement remains
when the copula training dates are held fixed. Online
Section~\ref{supp:refit-comparisons} gives all event comparisons,
coverage, interval lengths, and paired Monte Carlo uncertainty.

\begin{table}[!htbp]
\centering\small
\caption{Downside Probability Accuracy: Dated Transforms and Same-Family Refits}\label{tab:sim-benchmark-rmse}
\begin{tabularx}{\linewidth}{@{}rrRRRR@{}}
\toprule
$r$ & $T$ & Archive & Matched refit & Full refit & Full minus archive\\
\midrule
\multicolumn{6}{@{\extracolsep{\fill}}l}{\textit{Panel A: Gaussian copula}}\\
0.25 & 250 & 1.026 & 1.047 & 1.037 & 0.011 (0.009)\\
0.25 & 1000 & 0.527 & 0.523 & 0.522 & -0.005 (0.003)\\
0.25 & 4000 & 0.262 & 0.260 & 0.260 & -0.002 (0.001)\\
0.75 & 250 & 1.454 & 1.406 & 1.401 & -0.053 (0.013)\\
0.75 & 1000 & 0.718 & 0.719 & 0.718 & -0.001 (0.004)\\
0.75 & 4000 & 0.347 & 0.343 & 0.343 & -0.004 (0.002)\\
\midrule
\multicolumn{6}{@{\extracolsep{\fill}}l}{\textit{Panel B: Clayton copula}}\\
0.25 & 250 & 1.437 & 1.421 & 1.410 & -0.027 (0.015)\\
0.25 & 1000 & 0.688 & 0.670 & 0.663 & -0.025 (0.005)\\
0.25 & 4000 & 0.345 & 0.339 & 0.336 & -0.009 (0.002)\\
0.75 & 250 & 1.584 & 1.523 & 1.528 & -0.056 (0.017)\\
0.75 & 1000 & 0.704 & 0.704 & 0.703 & -0.001 (0.005)\\
0.75 & 4000 & 0.443 & 0.440 & 0.440 & -0.003 (0.002)\\
\bottomrule
\end{tabularx}
\par\smallskip
\begin{minipage}{\linewidth}\footnotesize\textit{Note:} Entries are root mean squared error (RMSE) multiplied by 100. Lower values indicate greater probability accuracy. Archive is the proposed estimator. Matched refit recomputes residuals at the terminal marginal estimates and fits the same copula family on the original $T-b_T$ dates. Full refit uses all $T$ dates. All margins use the same observations. The last column gives the paired RMSE difference and its Monte Carlo standard error in parentheses, so a negative difference favors the full refit. Each row uses 2,000 common replications. Downside requires all three outcomes to be at most $-1$. Online Section~\ref{supp:refit-comparisons} reports both refit differences, the other events, coverage, and interval length.
\end{minipage}
\end{table}

The larger gains under Clayton innovations arise from the copula
family. Clayton fits have lower RMSE than Gaussian fits in all 54
combinations of design, event, and timing rule.
For Downside at $r=0.75$, $T=4000$, a full Clayton refit has
RMSE 0.00440 and coverage 94.5\%, compared with 0.01162 and
about 12\% for the full Gaussian refit. The dated and matched-date
comparisons give the same substantive conclusion. Allowing Clayton
dependence improves forecasts of this joint downside event even when
all historical transforms are recomputed. Table~\ref{tab:family-contrast}
online isolates this family effect from the timing choice.

The stationary-vine point estimates are slightly more accurate than
Archive in 17 of the 18 Gaussian comparisons across designs and events. For
Downside at $r=0.75$, $T=250$, RMSE is 0.01383 for the stationary
vine and 0.01454 for Archive. The paired difference is $-0.00071$
with Monte Carlo standard error 0.00018. Both stationary-vine bootstrap
versions use this same point forecast, so their accuracy rankings
coincide.

Their interval comparisons give different conclusions
(Table~\ref{tab:sim-benchmark-coverage}). At $r=0.75$, $T=250$,
Downside coverage is 94.1\% for Archive, 90.8\% for S-vine J, and
95.2\% for S-vine N. The native bootstrap therefore has coverage
closer to 95\% than Archive, with a shorter mean interval, 0.0243
against 0.0279. Its coverage exceeds Archive by 1.05 percentage
points, with paired Monte Carlo standard error 0.53 points.
Archive instead improves on the full-Jacobian bootstrap in this
comparison, whose mean interval length is 0.0224.

At $r=0.25$, native Downside coverage is 89.4\%--90.0\%, compared
with 94.9\%--95.1\% for Archive. At $r=0.75$, its coverage rises
to 96.8\% and 97.3\% at $T=1000$ and $T=4000$, where Archive is
closer to nominal. The proposed method thus gives better calibration
in these designs but does not uniformly outperform stationary-vine
inference. The two bootstrap algorithms differ in their derivative
matrix, multiplier distribution, and block rule. Their differences
cannot be attributed to the Jacobian alone. All 999 draws are valid
for both algorithms in every replication. Online
Section~\ref{supp:sim-benchmark} reports the Mixed event and paired
Monte Carlo uncertainty for all comparisons.

\begin{table}[!htbp]
\centering\small
\caption{Gaussian Designs: Coverage and Precision against Stationary-Vine Benchmarks}\label{tab:sim-benchmark-coverage}
\begin{tabularx}{\linewidth}{@{}rrRRRR@{}}
\toprule
$r$ & $T$ & Archive & Full refit & S-vine J & S-vine N\\
\midrule
\multicolumn{6}{@{\extracolsep{\fill}}l}{\textit{Panel A: Zero event}}\\
0.25 & 250 & 95.5 [0.1437] & 94.7 [0.1375] & 93.8 [0.1347] & 91.8 [0.1263]\\
0.25 & 1000 & 94.8 [0.0710] & 95.0 [0.0696] & 94.4 [0.0686] & 91.9 [0.0641]\\
0.25 & 4000 & 95.5 [0.0350] & 95.1 [0.0345] & 94.7 [0.0342] & 92.7 [0.0316]\\
0.75 & 250 & 95.0 [0.1016] & 94.9 [0.0984] & 91.6 [0.0919] & 95.2 [0.1029]\\
0.75 & 1000 & 94.9 [0.0486] & 94.4 [0.0477] & 92.5 [0.0447] & 96.2 [0.0506]\\
0.75 & 4000 & 94.5 [0.0244] & 94.3 [0.0241] & 92.4 [0.0229] & 96.3 [0.0256]\\
\midrule
\multicolumn{6}{@{\extracolsep{\fill}}l}{\textit{Panel B: Downside event}}\\
0.25 & 250 & 94.9 [0.0384] & 94.8 [0.0372] & 93.4 [0.0344] & 89.4 [0.0293]\\
0.25 & 1000 & 94.9 [0.0184] & 95.0 [0.0181] & 93.8 [0.0176] & 89.6 [0.0149]\\
0.25 & 4000 & 95.1 [0.0091] & 95.2 [0.0090] & 94.7 [0.0088] & 90.0 [0.0074]\\
0.75 & 250 & 94.1 [0.0279] & 94.4 [0.0260] & 90.8 [0.0224] & 95.2 [0.0243]\\
0.75 & 1000 & 94.9 [0.0132] & 94.7 [0.0129] & 92.8 [0.0118] & 96.8 [0.0124]\\
0.75 & 4000 & 95.4 [0.0064] & 95.0 [0.0063] & 93.6 [0.0059] & 97.3 [0.0061]\\
\bottomrule
\end{tabularx}
\par\smallskip
\begin{minipage}{\linewidth}\footnotesize\textit{Note:} Coverage is in percent for nominal 95\% intervals, with mean length in brackets. Coverage closer to 95\% indicates better calibration. Archive and the same-family Full refit use logit intervals. S-vine is the Gaussian stationary vine of \citet{nagler_kruger_min_2022}, which models serial and cross-variable dependence directly in levels. J uses the full score Jacobian and Gaussian multipliers. N uses the native \texttt{svines} bootstrap, including its Hessian and multiplier rule. Both use the same corrected point fit, 999 original draws, valid-draw filtering, and percentile endpoints. These variants assess the sensitivity of inference for the same stationary-vine forecast. Each row uses 2,000 common replications. Unavailable intervals count as noncoverage, mean length uses available intervals, and a subscript denotes availability below 2,000. Online Section~\ref{supp:sim-benchmark} defines both algorithms and reports paired Monte Carlo uncertainty. Zero and Downside use thresholds zero and $-1$ for all three outcomes.
\end{minipage}
\end{table}

\subsection{Historical and current marginal uncertainty}\label{sec:sim-components}

To examine the additional variance from retaining dated transforms,
we simulate a bivariate independent and identically distributed
Gaussian process with correlation $\rho\in\{-0.5,0,0.5\}$, using
$T=1000,4000$ and 4,000 replications per design. Estimation still fits
an unrestricted bivariate VAR(1). The probability that both variables
are at most zero is $1/4+\arcsin(\rho)/(2\pi)$. We group results
by terminal-state leverage, $\lambda_x=1+y_T'R_Z^{-1}y_T$, which
measures how far the forecast state is from the mean after adjusting
for its covariance.

Figure~\ref{fig:sim-archive-mechanism} shows that estimated archive
variance shares follow the pattern in Corollary~\ref{rli:gaussian-zero}.
The additional variance is zero at independence, and its estimated
share moves toward zero as $T$ increases. At equal positive and
negative correlations, the additional variance is the same. Its share
is larger under negative correlation because current marginal-mean
uncertainty is smaller. Higher leverage reduces the share in both
cases. The proposed estimator has slightly higher empirical variance
than Gaussian refitting in these designs, as reported in Online
Table~\ref{tab:sim-mechanism-variance}. The variance formula accounts
for this measured cost of retaining dated transforms.

\begin{figure}[!htbp]
\centering
\includegraphics[width=\linewidth]{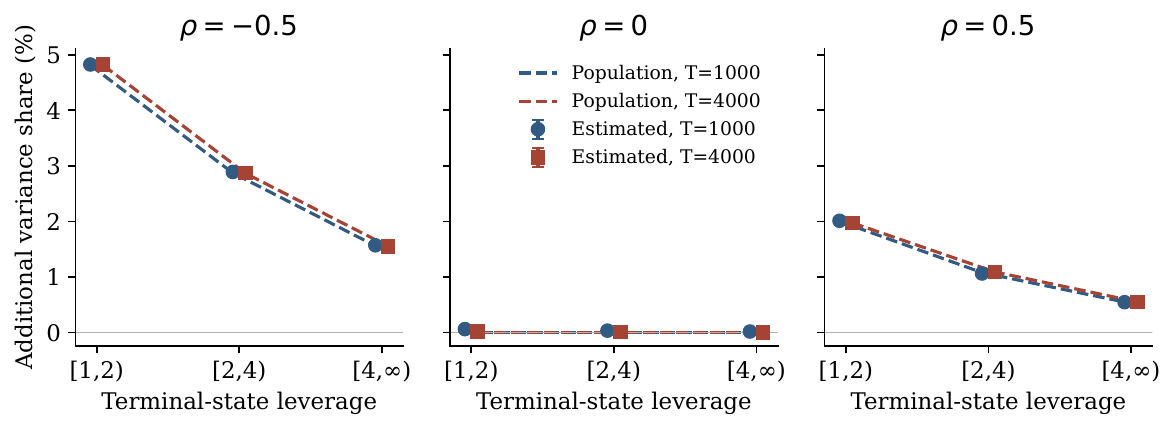}
\caption{The archive variance share across dependence and forecast states}
\label{fig:sim-archive-mechanism}
\par\smallskip\begin{minipage}{\linewidth}\footnotesize
\textit{Note:} Each series averages the replication-specific ratio of
additional archive variance to total archive variance within a fixed
leverage bin. Dashed lines use the population formulas evaluated at
the simulated states. Markers use fitted covariances, with bars equal
to 1.96 Monte Carlo standard errors of the within-bin mean. The
simulation uses 4,000 Monte Carlo replications per correlation and sample size.
The VAR intercept, slopes, scales, and pair parameter are all estimated.
\end{minipage}
\end{figure}

Accounting for current marginal estimation substantially improves
coverage (Table~\ref{tab:sim-components}). At $T=1000$ and $r=0.75$,
removing all current marginal uncertainty lowers Downside coverage
from 94.9\% to 50.6\% under Gaussian innovations and from 94.1\%
to 35.8\% under Clayton innovations. Removing only current slope
uncertainty gives 87.1\% and 88.3\%. The full covariance therefore
outperforms both omissions on coverage. This result supports the role
of the terminal marginal estimates in Theorem~\ref{main:dynamic}:
current slope estimation affects the forecast even though historical
slope errors average out of the leading copula expansion.

Removing only the additional archive variance gives Downside coverage
of 94.6\% and 93.9\% in the same designs. Its inclusion produces a
much smaller improvement. The supplement also omits the covariance
between copula and current marginal estimates, or the effect of earlier
pair estimates on later pairs. Those comparisons show small coverage
changes in either direction, with no consistent advantage from either
omission. The largest gains come from retaining current marginal
uncertainty.

\begin{table}[!htbp]
\centering
\caption{Coverage when sources of estimation uncertainty are omitted}\label{tab:sim-components}
\small
\setlength{\tabcolsep}{4pt}
\renewcommand{\arraystretch}{1.06}
\begin{tabularx}{\linewidth}{@{\extracolsep{\fill}}lrrrrr}
\toprule
Copula & Event & Full & No history & No slopes & No terminal\\
\midrule
Gaussian & Zero & 95.1 [0.049] & 94.9 [0.049] & 80.2 [0.031] & 38.3 [0.011]\\
Gaussian & Downside & 94.9 [0.013] & 94.6 [0.013] & 87.1 [0.008] & 50.6 [0.003]\\
Gaussian & Mixed & 94.5 [0.041] & 94.2 [0.041] & 80.3 [0.026] & 42.6 [0.008]\\
Clayton & Zero & 94.8 [0.048] & 94.6 [0.048] & 82.2 [0.032] & 35.4 [0.011]\\
Clayton & Downside & 94.1 [0.015] & 93.9 [0.015] & 88.3 [0.010] & 35.8 [0.004]\\
Clayton & Mixed & 94.5 [0.041] & 94.0 [0.040] & 79.5 [0.025] & 40.3 [0.008]\\
\bottomrule
\end{tabularx}
\par\smallskip\begin{minipage}{\linewidth}\footnotesize
\textit{Note:} Coverage is in percent for nominal 95\% intervals, with mean length in brackets. Coverage closer to 95\% is better, and shorter intervals indicate greater precision at comparable coverage. Unavailable intervals count as noncoverage, while mean length uses available intervals. All columns use the same proposed point estimate and logit transformation at $r=0.75$, $T=1000$, with 4,000 replications per row. Full includes every covariance component. No history removes the additional archive variance. No slopes removes uncertainty in current slopes. No terminal removes uncertainty in all current marginal coefficients and scales. These comparisons measure the contribution of each source of estimation uncertainty. Zero and Downside require variables 1--3 to be at most zero and $-1$, respectively. Mixed requires $y_{1,T+1}\leq0$ and $y_{3,T+1}>0$, with variable 2 unrestricted.
\end{minipage}
\end{table}

\FloatBarrier
\subsection{Additional simulation results}

With nonzero cross-variable lag coefficients and with five variables,
logit coverage for Zero and Downside is 94.1\%--95.5\%
(Online Table~\ref{tab:sim-extensions}). Logit intervals generally
improve on Wald intervals, particularly for Downside, while finite
training weights produce almost identical results.

Under weak Clayton dependence, Downside logit coverage is 62.8\% at
$T=250$ and 92.9\% at $T=1000$, compared with Wald coverage of
58.1\% and 87.1\% (Online Table~\ref{tab:sim-stress}). Intervals
are available in 64.6\% and 97.8\% of replications, respectively,
and coverage uses all replications. Logit intervals improve coverage
in this comparison. With standardized Student $t_5$ innovations and
fitted normal margins, logit coverage is 58.9\% at $T=1000$ and
33.3\% at $T=4000$, compared with 55.7\% and 33.8\% for Wald.
Neither interval construction consistently outperforms the other in
the Student design.

The twenty- and fifty-variable simulations in Online
Section~\ref{supp:scalability} assess larger systems. Under Clayton
dependence, the proposed model has lower event RMSE and coverage
closer to 95\% than the Gaussian refit. In the Gaussian designs,
event accuracy is similar to the fixed stationary vine, which performs
better at fifty variables and 4,000 observations. The proposed model
has lower full-system density loss than the fixed and sparse stationary
vines in all four matched Gaussian designs. Same-family refitting has
slightly lower density loss still. All fifty-variable fits and analytic
intervals complete. The timing comparisons assess computational
feasibility for these fixed-depth vines, with the dense Gaussian VAR
remaining faster for point forecasting.

\FloatBarrier

\providecommand{\calibrationproposition}{Proposition~\ref{cal:proposition}}

\section{Empirical Application: Real-time Forecasting of Joint Contractions}\label{sec:evaluation}

We forecast joint contractions in U.S. employment and production using
fourteen macroeconomic variables. The application illustrates the
paper's Gaussian special case: dated marginal forecasts are combined
through a fixed vine, and their estimation error is carried into an
interval for the current event probability. Nested three- and
eleven-variable systems assess the predictive contribution of a broader
information set. All comparisons use the same 231 monthly forecast
origins and locked release outcomes.

\subsection{Economic information and the release target}

The application builds on the macroeconomic copula-VAR setting of
\citet{tsionas_izzeldin_trapani_2022}, with a different information set,
release protocol, and target. Payroll employment, industrial production,
and unemployment form the three-variable core. Manufacturing hours,
labor-force participation, manufacturing employment, and capacity
utilization measure labor input and capacity pressure. Consumer inflation,
policy-rate changes, the term spread, and real money growth complete the
eleven-variable system. Real business-credit growth, core inflation, and
the dollar--sterling exchange rate add credit, underlying price pressure,
and external-price information to the fourteen-variable system.

The dimension follows economic roles and usable vintage histories.
Variables with inadequate initialization histories, discontinued series,
or conflicts with the forecast release calendar are excluded. We do not
replace them with modern splices or retrospectively constructed histories.
The retained system consequently omits housing and household expenditure.
Its selection precedes forecast scoring and does not establish that every
included variable is necessary. Online Table~\ref{emp:panel} gives the
series, transformations, and nesting.

Historical vintages come from ALFRED (Archival FRED), maintained by the
Federal Reserve Bank of St.~Louis.\footnote{\url{https://alfred.stlouisfed.org/}.}
The information available at an issue date and the release defining an
outcome must be distinguished
\citep{croushore_stark_2001,koenig_dolmas_piger_2003}.
Let $\tau_m$ denote the first vintage containing the current and previous
levels of all three core series for reference month $m$. Define
\begin{align*}
 E_m&=100\log\frac{\mathrm{PAYEMS}_m^{[\tau_m]}}
                        {\mathrm{PAYEMS}_{m-1}^{[\tau_m]}},&
 I_m&=100\log\frac{\mathrm{INDPRO}_m^{[\tau_m]}}
                        {\mathrm{INDPRO}_{m-1}^{[\tau_m]}},\\
 U_m&=\mathrm{UNRATE}_m^{[\tau_m]}-\mathrm{UNRATE}_{m-1}^{[\tau_m]}.
\end{align*}
The core record $(10E_m,I_m,-10U_m)'$ is locked at $\tau_m$.
This is a first common release, so a component published earlier may
already have been revised. Additional predictors are recorded using
vintages available at the common panel issue date. That date precedes
the next core outcome release. The same panel issue dates are used for
all three dimensions, even when a smaller system could be issued earlier.
These are retrospective reconstructions of a specified information set.
Locking an observation at its release and retaining the marginal forecast
used to transform it are separate choices. The terminal-refit benchmark
keeps the same locked observations but recalculates their transforms.
This comparison changes marginal estimation timing without introducing
later data revisions.

The initial vintage, 1 July 1997, supplies 240 monthly regression
observations from June 1977 through May 1997 and one presample lag.
After 100 subsequent archive observations, the first scored target is
October 2005. The primary evaluation contains 171 targets through
December 2019. All 60 targets from January 2020 through December 2024
form a separate stress period. Training samples grow from 340 to 570
observations. The headline event is
$A_{EP}=\{E<0,I<0\}$. It avoids the rounded unemployment boundary while
retaining unemployment as a predictor. Joint labor contraction,
three-variable contraction, and a more severe downside event provide
supporting comparisons online.

\subsection{Forecast construction and benchmark comparisons}

The main specification is an unrestricted VAR(1) with an intercept,
normal innovation margins, and a Gaussian D-vine in the order of
Online Table~\ref{emp:panel}, truncated after tree two. Every retained
pair and every VAR coefficient is estimated. Gaussian pairs give a
structured Gaussian innovation distribution. The additional predictors
affect the event through its conditional marginal forecasts. Historical PITs retain
the marginal fit preceding each outcome. Terminal margins use the
current sample. Dependence estimation starts after
$b_T=\max\{\lceil T^{0.7}\rceil,5(n+1),240\}$ observations.
The safeguards center and scale regressors using the initial sample,
then freeze those constants. They preserve the OLS objective and
original outcome units whenever inactive. Online
Section~\ref{emp:safeguards} states the revised rules and proves their
asymptotic equivalence to the original safeguards under the maintained
fixed-dimension assumptions. No retained Gram or variance safeguard
activates in the fourteen-variable primary forecasts.

The Gaussian refit recomputes historical transforms from terminal
margins on the same locked observations. Independence keeps those
margins and removes innovation dependence. A dense Gaussian VAR and a
shrinkage VAR supply conventional forecasting benchmarks. Fixed and
sparse stationary vines model serial and cross-sectional dependence
together, following the specification class of
\citet{nagler_kruger_min_2022}. The fixed benchmark uses Gaussian pairs.
The sparse benchmark selects Gaussian or independence pairs. Both
use full-range Gaussian pair optimization on the complete forecast
calendar. Their dynamic parameterization differs from a VAR with an
innovation vine, so their comparisons do not isolate the effect of
retaining dated transforms.

We evaluate probabilities by Brier loss, the squared difference from
the recorded binary outcome \citep{gneiting_raftery_2007}.
Table~\ref{emp:alfred:scores} reports the common-calendar comparisons.
\begin{table}[!htbp]
\centering\small
\caption{Joint Employment--Production Contraction: Brier Loss}
\label{emp:alfred:scores}
\begin{tabularx}{\linewidth}{@{}Lrrr@{}}
\toprule
Method & Variables & 2005:10--2019:12 & 2020:01--2024:12 \\
\midrule
Dated-transform vine & 3 & 0.0495 & 0.1059 \\
Dated-transform vine & 11 & 0.0465 & 0.1259 \\
Dated-transform vine & 14 & 0.0446 & 0.1146 \\
Gaussian refit / dense Gaussian VAR & 14 & 0.0454 & 0.1424 \\
Independent innovations & 14 & 0.0408 & 0.1019 \\
Shrinkage VAR & 14 & 0.0482 & 0.1355 \\
Fixed stationary vine & 14 & 0.0455 & 0.1420 \\
Sparse stationary vine & 14 & 0.0471 & 0.0940 \\
\bottomrule
\end{tabularx}
\par\smallskip
\begin{minipage}{\linewidth}\footnotesize\textit{Note:} Mean Brier loss, the squared difference between the event probability and binary outcome, over 171 primary and 60 stress forecasts. Lower loss is better. All methods use the same issue dates and release outcomes, with complete point-forecast availability. Dated-transform vine is the proposed method. Gaussian refit recomputes past transforms at current marginal estimates, and the dense Gaussian vector autoregression (VAR) estimates an unrestricted innovation covariance. Independence removes innovation dependence, and shrinkage reduces coefficient estimation variability through a normal--inverse-Wishart prior. The stationary vines of \citet{nagler_kruger_min_2022} model serial and cross-variable dependence jointly, with Gaussian pairs fixed in one specification and selected against independence in the sparse specification. Gaussian refit and the dense Gaussian VAR give identical probabilities for this event. The nested dimensions compare additional forecasting information. Implementation details are in Online Section~\ref{emp:benchmarks}.
\end{minipage}
\end{table}

Before 2020, the dated-transform loss falls from $0.0495$ with three
variables to $0.0465$ with eleven and $0.0446$ with fourteen. The full
system has lower observed loss than Gaussian refitting, shrinkage, and
either stationary vine. Paired monthly block comparisons qualify these
rankings. The fourteen-minus-three-variable loss difference is
$-0.0049$, with Holm-adjusted $p=0.0520$ across the three prespecified
primary comparisons. The smaller improvement over Gaussian refitting,
$-0.0007$, has adjusted $p=0.0495$. The shrinkage comparison has
adjusted $p=0.0520$. Neither stationary-vine contrast excludes zero
in its supplementary unadjusted 95\% interval. These results support
competitive performance, with limited evidence of an incremental gain
from the broader information set.

The stress period changes the ranking. The fourteen-variable archive
has loss $0.1146$, below the fixed stationary vine's $0.1420$ but above
the sparse vine's $0.0940$. The three-variable archive also has lower
loss than the fourteen-variable specification. Independence outperforms
the latter in both periods. Thus neither additional variables nor
estimated innovation dependence improves this event forecast uniformly.
Online Section~\ref{emp:benchmarks} reports the paired uncertainty,
other events, and further benchmarks. The stochastic-volatility VAR
supplies only 70 primary forecasts and no stress forecasts under its
computing and convergence limits, so its available-date score cannot
be ranked as a full-calendar comparison.

\subsection{Estimation uncertainty and computational scalability}

The analytic covariance accounts jointly for dated transforms, terminal
margins, and their cross covariances. For employment--production
contraction in the primary period, the additional archive block has a
median share of $2.32\%$ of the limiting-weight event variance. Removing
all terminal-margin gradients instead reduces the median standard-error
ratio to $31.1\%$ of the complete finite-weight value. These calculations
hold the fitted probability fixed. They illustrate why the theoretical
archive cost can be modest even when current marginal estimation
matters substantially. The Gaussian covariance ordering concerns a
common population law and does not impose a finite-sample Brier ranking.

All 171 primary employment--production intervals are available, with
mean length $0.1454$. During 2020--2024, 59 of 60 are available and
their median length is $0.9999$, so those intervals provide little
precision about the event probability. The application evaluates
forecasts against observed outcomes. Coverage of the unobserved true
conditional probability is assessed in the simulations, rather than
inferred from empirical interval availability. Online
Section~\ref{emp:uncertainty} reports residual dispersion and dependence
alongside the uncertainty calculations.

Computation remains practical at fourteen variables. On twelve selected
dates, the median recurring point update takes $0.0027$ seconds, followed
by $0.0123$ seconds for analytic inference. A full-history archive fit
takes $0.0418$ seconds. Corrected fixed and sparse stationary-vine point
fits take $0.3600$ and $0.3700$ seconds. Their fixed-family bootstrap
procedures produce intervals at all twelve checked dates, with median
inference times of $85.97$ and $94.26$ seconds. These are different
inferential procedures, and the comparison does not establish equal
coverage. The dense Gaussian VAR remains faster for point forecasting,
at $0.0009$ seconds. Online Section~\ref{emp:computing} separates initial
fitting, recurring updates, and inference costs.

The application demonstrates inexpensive sequential updates and analytic
uncertainty calculation in an economically broader system. The empirical
events directly restrict two or three innovation coordinates, while the
additional variables enter their conditional means. The twenty- and
fifty-variable simulations in Online Section~\ref{supp:scalability}
assess the same computational tasks at longer sample sizes. Together,
the results support practical scalability for fixed-depth vines and
these event targets. The forecast rankings determine when the broader
information set improves prediction.

\section{Conclusion}\label{sec:conclusion}

Inference for joint-event probabilities must track both the dated
marginal fits used to estimate dependence and the terminal fit used
to forecast. In the Vine Copula VAR, averaging reduces the historical
regression correction to harmonically weighted innovation moments.
The terminal forecast retains uncertainty in every coefficient.
Their joint influence yields a feasible covariance estimator and
repeated-sample intervals at the random forecast state. Fixed one-sided
events have positive limiting variance under the maintained Gaussian
and positive Clayton specifications.

The same-family terminal-refit comparison quantifies the variance cost
of retaining dated transforms. In the bivariate Gaussian specialization,
dated scale sensitivity generates this cost, which vanishes at
independence. Its share can remain small away from independence when
terminal-mean uncertainty is substantial. The simulations confirm a
modest accuracy cost of retaining the transforms. Same-family refits
usually have lower probability RMSE and always have shorter mean
intervals in the core comparisons. Clayton fits outperform Gaussian
fits under Clayton dependence with either transform rule. The separate
covariance-omission comparisons show that current marginal uncertainty
accounts for the largest coverage gains. Stationary-vine coverage
depends materially on the bootstrap construction and persistence.
These findings assign distinct roles to the timing correction, the
dependence specification, and the terminal marginal estimates.

The fourteen-variable application implements the Gaussian special case
with economically broader information and inexpensive recurring
forecasts. Its observed pre-2020 employment--production score improves
on the three-variable model, while independence has lower loss in both
evaluation periods and the sparse stationary vine performs better
during 2020--2024. The twenty- and fifty-variable simulations provide
complementary evidence on feasible estimation and analytic inference
for fixed-depth vines and specified events. Together, these results
show how the joint covariance calculation supports uncertainty
assessment for a forecasting rule that preserves its dated marginal
inputs.




\clearpage
\appendix
\numberwithin{table}{section}
\begingroup
\setstretch{1.15}

{\noindent\LARGE\bfseries Appendices\par}
\section{Joint inference for the copula and terminal VAR estimates}
\label{dji:section}

We expand the historical score correction and terminal regression
errors on the same probability space. The recursive score expansion
first removes the centered lag-design contribution. We then calculate
the covariance with the terminal influences and prove the joint CLT
by finite-lag approximation. Online Appendix~\ref{supp:joint-proofs}
provides the bounds for recursive estimation and moment-consistency results.

Throughout this appendix, Assumptions~\ref{main:var-assumption}--\ref{main:regularity}
and the estimators in Section~\ref{sec:method} are maintained. Write
$N_T=T-b_T$, $A_Q=Q^{-1}$, $\varphi_t=\varphi(Z_t)$,
$w_t=(w_{1,t}',\ldots,w_{n,t}')'$, and
$L_\beta=(M_{\mu,1}e_1',\ldots,M_{\mu,n}e_1')$.
The terminal errors and influences are those in
\eqref{main:terminal-influences}. The covariance is
\eqref{main:joint-covariance}. Norms are Euclidean or operator norms.
Constants may depend on the fixed model but not on $T$.

\begin{equation}
 U_T=\frac1{\sqrt T}\sum_{t=1}^T
 \begin{pmatrix}
 -J^{-1}\{\mathbf1\{t>b_T\}g(\theta_0,Z_t)+h_{T,t}B_g\varphi_t\}\\ e_t
 \end{pmatrix},\qquad
 h_{T,t}=\sum_{m=\max(b_T,t)}^{T-1}m^{-1}.
 \label{dji:joint-expansion}
\end{equation}
Our task is to show that $U_T$ is the leading joint influence, including
its off-diagonal covariance blocks. The coefficients $h_{T,t}$ encode
the repeated use of observations in the dated regression fits.

\subsection{Recursive score expansion}\label{dyn:section}

Historical location error is the product of a regressor vector and a recursive
regression error. Both depend on past innovations, so replacing the
regressor vector by its expectation requires justification. The next lemma
shows that the centered contribution vanishes after aggregation,
leaving innovation moments in the historical correction.

Write $\mu_x=Ex_t$, $M_{i,m}=\sum_{s\le m}x_sZ_{is}$ and $s_{i,m}=M_{i,m}/m$.
Lemmas~\ref{dyn:moments}--\ref{dyn:first-stage} establish the bounds for
recursive estimation and the predictable localization used below.
For $t=m+1$, the nuisance vector in the next statement is
$a_t=((x_t'd_{i,m})_{i\le n}',(a^v_{i,m})_{i\le n}')'$, with
$d_{i,m}=(\widehat\beta_{i,m}-\beta_i)/\sigma_i$ and
$a^v_{i,m}=\widehat v_{i,m}/\sigma_i^2-1$, as in
Lemma~\ref{dyn:first-stage}. Its coordinates are the standardized
conditional-mean errors and relative-variance errors used to form the
dated residual at $t$.

\begin{lemma}[Centered lag-design cancellation]
\label{dyn:interaction}
Under the conditions of Lemma~\ref{dyn:moments}, for each $i\le n$,
\begin{equation}
 \frac{1}{\sqrt T}\sum_{m=b_T}^{T-1}
       (x_{m+1}-\mu_x)'Q^{-1}s_{i,m}=o_p(1).
 \label{dyn:centered-interaction}
\end{equation}
Consequently the actual marginal nuisance vectors satisfy
\begin{equation}
 \frac{1}{\sqrt T}\sum_{t>b_T}a_t
 =\frac{1}{\sqrt T}\sum_{s=1}^T h_{T,s}\varphi(Z_s)+o_p(1).
 \label{dyn:averaged-nuisance}
\end{equation}
\end{lemma}

\begin{proof}
The centered design is correlated with the accumulated regression error.
The stable filter separates that dependence into a current-innovation
term and a telescoping difference. Put $D_j=C_{j+1}$, $H=\sum_{j\ge0}D_j$,
$V_j=\sum_{k>j}D_k$, and
$v_m=\sum_{j\ge0}V_j\varepsilon_{m-j}$. The geometric coefficient bound
gives finite moments of every order for $v_m$, uniformly in $m$.
Indeed, $H-V_0=D_0$ and $V_{j-1}-V_j=D_j$ for $j\ge1$, so
\[
 H\varepsilon_m+v_{m-1}-v_m
 =(H-V_0)\varepsilon_m+
       \sum_{j\ge1}(V_{j-1}-V_j)\varepsilon_{m-j}
 =\sum_{j\ge0}D_j\varepsilon_{m-j}.
\]
This proves the exact identity
\begin{equation}
 x_{m+1}-\mu_x=H\varepsilon_m+v_{m-1}-v_m.
 \label{dyn:coboundary}
\end{equation}
Write $M_m=M_{i,m}$ and $K=Q^{-1}$. The contribution of $H\varepsilon_m$
to the unnormalized sum in \eqref{dyn:centered-interaction} is
\begin{align}
 \mathsf{U}_T+\mathsf{D}_T
 &=\sum_{m=b}^{T-1}\frac{(H\varepsilon_m)'K M_{m-1}}{m}
   +\sum_{m=b}^{T-1}\frac{(H\varepsilon_m)'K x_m Z_{im}}{m}.
 \label{dyn:interaction-split}
\end{align}
The summands in $\mathsf{U}_T$ are martingale differences, because $M_{m-1}$ is
$\mathcal F_{m-1}$-measurable and $E\varepsilon_m=0$. Their variances
satisfy
\begin{align}
 E|\mathsf{U}_T|^2
 &\le C\sum_{m=b}^{T-1}m^{-2}E\|M_{m-1}\|^2
 \le C\sum_{m=b}^{T-1}m^{-1}\le C(1+\log T),\\
 E|\mathsf{D}_T|
 &\le C\sum_{m=b}^{T-1}m^{-1}
                 E(\|\varepsilon_m\||Z_{im}|)E\|x_m\|
 \le C(1+\log T).
 \label{dyn:interaction-leading-bounds}
\end{align}
The factorization in the second line uses independence of the current
innovation vector from the predictable design, and does not separate
coordinates of that innovation vector.

The telescoping difference can be transferred to the slowly changing
recursive sample average. Set $w_m=K M_m/m$. Summation by parts gives
\begin{align}
 \sum_{m=b}^{T-1}(v_{m-1}-v_m)'w_m
 &=v_{b-1}'w_b-v_{T-1}'w_{T-1}
       +\sum_{m=b}^{T-2}v_m'(w_{m+1}-w_m),\\
 w_{m+1}-w_m
 &=K\left\{\frac{x_{m+1}Z_{i,m+1}}{m+1}
                       -\frac{M_m}{m(m+1)}\right\}.
 \label{dyn:summation-by-parts}
\end{align}
No independence is needed for the following bounds. The endpoint bounds
follow by Cauchy--Schwarz and $E\|M_m\|^2\le Cm$:
\begin{align}
 E|v_{b-1}'w_b|+E|v_{T-1}'w_{T-1}|
 &\le C(b^{-1/2}+T^{-1/2}),\\
 \sum_{m=b}^{T-2}\frac{E|v_m'K M_m|}{m(m+1)}
 &\le C\sum_{m=b}^{T-2}m^{-3/2}\le Cb^{-1/2}.
\end{align}
Both $v_m$ and $x_{m+1}$ are $\mathcal F_m$-measurable. Their finite
fourth moments and independence of $Z_{i,m+1}$ from $\mathcal F_m$ give
\begin{align}
 \sum_{m=b}^{T-2}\frac{E|v_m'K x_{m+1}Z_{i,m+1}|}{m+1}
 &\le C\sum_{m=b}^{T-2}(m+1)^{-1}
                \{E\|v_m\|^2E\|x_{m+1}\|^2\}^{1/2}E|Z_{i1}|\\
 &\le C(1+\log T).
 \label{dyn:coboundary-bound}
\end{align}
Combining these displays proves that the left side of
\eqref{dyn:centered-interaction} is
$O_p\{\sqrt{\log T/T}+(1+\log T)/\sqrt T\}=o_p(1)$.

To apply this result to the estimated conditional means, use
\eqref{dyn:beta-expansion}. On $\mathcal E_T$ their normalized remainder
is bounded by
\[
 \frac{L\sqrt{\log T}}{\sqrt T}
      \sum_{m=b}^{T-1}\|r_{i,m}^{\beta}\|.
\]
Its expectation on the recursive estimation events is at most
$C(\log T)^{3/2}/\sqrt T$ by \eqref{dyn:remainder-moments}. Hence the
bound is $o_p(1)$ by Markov's inequality and
$P(\mathcal E_T^c)\to0$. Since the first coordinate of $x_s$ is one,
\begin{equation}
 Qe_1=E(x_tx_t'e_1)=Ex_t=\mu_x,
 \qquad \mu_x'Q^{-1}x_s=e_1'x_s=1.
 \label{dyn:intercept-identity}
\end{equation}
Thus \eqref{dyn:centered-interaction} gives
\begin{align}
 \frac{1}{\sqrt T}\sum_{t>b}a^\mu_{it}
 &=\frac{1}{\sqrt T}\sum_{m=b}^{T-1}\mu_x'Q^{-1}s_{i,m}+o_p(1)\\
 &=\frac{1}{\sqrt T}\sum_{m=b}^{T-1}\frac{1}{m}
                          \sum_{s=1}^m Z_{is}+o_p(1).
\end{align}
For the variance components, \eqref{dyn:variance-expansion} and
\eqref{dyn:remainder-moments} give a normalized remainder with expectation
at most $C(1+\log T)/\sqrt T=o(1)$ on the events controlling the recursive estimates. Reversing the
finite sums for the mean and variance components proves
\eqref{dyn:averaged-nuisance}. The resulting harmonic weights are exactly
$h_{T,s}=\sum_{m=\max(b_T,s)}^{T-1}m^{-1}$.
\end{proof}

\begin{proposition}[Recursive VAR expansion]
\label{dyn:theorem}
Under the model and estimator conditions in Section~\ref{sec:method}
and Assumptions~\ref{main:identification}--\ref{main:regularity}, the sequential estimator
\eqref{eq:sequential-fit} is consistent and satisfies
\begin{equation}
 \sqrt T(\widehat\theta-\theta_0)
 =-J^{-1}\frac1{\sqrt T}\sum_{s=1}^T
 \left[\mathbf1\{s>b_T\}g(\theta_0,Z_s)
       +h_{T,s}B_g\varphi(Z_s)\right]+o_p(1)
 \ \Longrightarrow N(0,V_{\rm seq}).
 \label{dyn:influence}
\end{equation}
Here the score moments and covariance are those in
\eqref{main:score-matrices} and \eqref{main:sequential-covariance}.
The upper-left block of \eqref{main:joint-feasible} consistently
estimates $V_{\rm seq}$. Every fixed contrast with positive limiting
variance has an asymptotically standard normal studentized statistic.
\end{proposition}

\begin{proof}
Put $N_T=T-b_T$. We first prove consistency to place the estimate in a
neighborhood with an invertible Jacobian and controlled generated-margin
remainders. We then expand the exact sequential score.

\paragraph{Consistency of the sequential estimator.}
For $t=m+1$, the standardized residual has the exact representation
\[
 \widehat Z_{it}=(Z_{it}-a^\mu_{it})(1+a^v_{i,m})^{-1/2}.
\]
On $\mathcal E_T$, Lemma~\ref{dyn:first-stage} makes $\|a_t\|$ uniformly
small. The mean-value theorem therefore implies
\begin{equation}
 \|\widehat Z_t-Z_t\|\le C(1+\|Z_t\|)\|a_t\|.
 \label{dyn:z-perturbation}
\end{equation}
The derivative envelopes give a fixed polynomial $W$ such that the
uniform edge-criterion perturbation is bounded on this event by
$N_T^{-1}\sum_{t>b_T}W(Z_t)\|a_t\|$.
Predictability of $a_t^*$ and independence of $Z_t$ from the past imply
\begin{align}
 E\left\{\frac{1}{N_T}\sum_{t>b_T}W(Z_t)\|a_t^*\|\right\}
 &=\frac{EW(Z_1)}{N_T}\sum_{t>b_T}E\|a_t^*\|\\
 &\le \frac{C\sqrt{\log T}}{N_T}
                  \sum_{m=b_T}^{T-1}m^{-1/2}
 =O\{\sqrt{\log T/T}\}=o(1).
 \label{dyn:criterion-bound}
\end{align}
All moments of $W(Z_1)$ exist because each coordinate is Gaussian and
the dimension is fixed. Markov's inequality and
$P(\mathcal E_T^c)\to0$ prove the required perturbation in probability.
Here is the uniform law and the separation argument explicitly.
Write $\mathbb P_Tf=N_T^{-1}\sum_{t>b_T}f(Z_t)$ and $Pf=Ef(Z_1)$.
For each sufficiently small $\delta>0$, choose a finite $\delta$-net
$\mathcal N_\delta$ of the compact parameter set. The derivative
envelope on its neighborhood gives
\begin{align*}
 \sup_\theta|(\mathbb P_T-P)\ell_e(\theta)|
 &\le\max_{v\in\mathcal N_\delta}|(\mathbb P_T-P)\ell_e(v)|
       +\delta(\mathbb P_TW+PW),\\
 P\left\{\max_{v\in\mathcal N_\delta}
       |(\mathbb P_T-P)\ell_e(v)|>\epsilon\right\}
 &\le\frac{|\mathcal N_\delta|\,PW^2}{N_T\epsilon^2}\longrightarrow0.
\end{align*}
Since $\mathbb P_TW\to_pPW$, first taking $T\to\infty$ and then
$\delta\downarrow0$ proves the uniform law used in
\citet[Lemma~2.4]{newey_mcfadden_1994}.
For $Q_e(v)=P\ell_e(\theta_{<e,0},v)$ and fixed $\epsilon>0$, compactness,
continuity, and uniqueness give
\[
 \Delta_e(\epsilon)=Q_e(\theta_{e,0})-
 \sup_{\|v-\theta_{e,0}\|\ge\epsilon}Q_e(v)>0
\]
whenever the set in the supremum is nonempty. Let
$\widehat Q_{e,T}(v)=N_T^{-1}\sum_{t>b_T}
\ell_e(\widehat\theta_{<e},v,\widehat Z_t)$ and
$u_{e,T}=\sup_v|\widehat Q_{e,T}(v)-Q_e(v)|$.
Maximization then implies
\begin{align*}
 \|\widehat\theta_e-\theta_{e,0}\|\ge\epsilon
 &\ \Longrightarrow\ 
 \Delta_e(\epsilon)\le Q_e(\theta_{e,0})-Q_e(\widehat\theta_e)\\
 &\hspace{34mm}\le2u_{e,T},\\
 P(\|\widehat\theta_e-\theta_{e,0}\|\ge\epsilon)
 &\le P\{2u_{e,T}\ge\Delta_e(\epsilon)\}.
\end{align*}
At the first stage there is no ancestor error. The uniform law and
\eqref{dyn:criterion-bound} give $u_{1,T}=o_p(1)$, hence
$\widehat\theta_1\to_p\theta_{1,0}$. This is the extremum-estimator
argument of \citet[Theorem~2.1]{newey_mcfadden_1994}.

Suppose the earlier blocks have been proved consistent. On a
neighborhood of their true values, the mean-value theorem gives
\begin{align}
 &\sup_{\theta_e}\left|\frac1{N_T}\sum_{t>b_T}
 \{\ell_e(\widehat\theta_{<e},\theta_e,Z_t)
       -\ell_e(\theta_{<e,0},\theta_e,Z_t)\}\right|\notag\\
 &\qquad\le\|\widehat\theta_{<e}-\theta_{<e,0}\|
       \frac1{N_T}\sum_{t>b_T}W(Z_t)
 =o_p(1)O_p(1)=o_p(1).
 \label{eq:sequential-ancestor-perturbation}
\end{align}
The true ancestor vector has a neighborhood inside the regularity
domain. Consistency makes the relevant parameter segments lie there
with probability tending to one. The same uniform law,
\eqref{dyn:criterion-bound}, and identification at each stage now prove
consistency at edge $e$. Induction ends after the fixed number of
edges. Because each true block is interior, all fitted maximizers
are interior with probability tending to one. Their exact
first-order conditions combine to give
\begin{equation}
 \frac1{N_T}\sum_{t>b_T}g(\widehat\theta,\widehat Z_t)=0.
 \label{eq:sequential-stacked-foc}
\end{equation}
Later parameter estimates do not change an earlier score.

\paragraph{Accumulated generated-score error.}
Put $g_t=g(\theta_0,Z_t)$ and
$A_t=A_g(\theta_0,Z_t)$ from \eqref{score:A-g}.
Taylor expansion in the $2n$ nuisance coordinates gives, on
$\mathcal E_T$,
\begin{equation}
 g(\theta_0,\widehat Z_t)=g_t+A_ta_t+R_t,
 \qquad \|R_t\|\le W_2(Z_t)\|a_t\|^2,
 \label{dyn:score-taylor}
\end{equation}
with another fixed polynomial envelope. Localization and
\eqref{dyn:localized-moments} give
\begin{align}
 E\left\{\frac{1}{\sqrt T}\sum_{t>b_T}W_2(Z_t)\|a_t^*\|^2\right\}
 &\le \frac{C\log T}{\sqrt T}\sum_{m=b_T}^{T-1}m^{-1}
 =O\{(\log T)^2/\sqrt T\}=o(1).
 \label{dyn:quadratic-score-bound}
\end{align}
Since $E(A_t\mid\mathcal F_{t-1})=B_g$, the variables
$(A_t-B_g)a_t^*$ are martingale differences. Orthogonality and current
innovation independence give
\begin{align}
 E\left\|\frac{1}{\sqrt T}\sum_{t>b_T}(A_t-B_g)a_t^*\right\|^2
 &\le \frac{C}{T}\sum_{t>b_T}E\|a_t^*\|^2\\
 &\le \frac{C\log T}{T}\sum_{m=b_T}^{T-1}m^{-1}
 =O\{(\log T)^2/T\}=o(1).
 \label{dyn:centered-score-bound}
\end{align}
On $\mathcal E_T$, the localization equals the actual nuisance vector.
Equations \eqref{dyn:score-taylor}--\eqref{dyn:centered-score-bound}
and Lemma~\ref{dyn:interaction} therefore yield
\begin{equation}
 \frac{1}{\sqrt T}\sum_{t>b_T}g(\theta_0,\widehat Z_t)
 =\frac{1}{\sqrt T}\sum_{s=1}^T
       [\mathbf1\{s>b_T\}g_s+h_{T,s}B_g\varphi(Z_s)]+o_p(1)
 =:L_T+o_p(1).
 \label{dyn:score-expansion}
\end{equation}
More explicitly, the remainder $r_T$ in this equality is
\begin{align*}
 r_T&=T^{-1/2}\sum_{t>b_T}R_t
 +T^{-1/2}\sum_{t>b_T}(A_t-B_g)a_t\\
 &\quad+B_g\left\{T^{-1/2}\sum_{t>b_T}a_t
            -T^{-1/2}\sum_{s=1}^Th_{T,s}\varphi(Z_s)\right\},\\
 \|r_T\|&=O_p\{(\log T)^2/\sqrt T\}
       +O_p(\log T/\sqrt T)
       +O_p\{(1+\log T)^{3/2}/\sqrt T\}=o_p(1).
\end{align*}
The first term uses Markov's inequality in
\eqref{dyn:quadratic-score-bound}, the second uses the second-moment
bound in \eqref{dyn:centered-score-bound}, and the third uses the
normalized bounds in Lemma~\ref{dyn:interaction}. Each calculation
holds on $\mathcal E_T$, whose complement has vanishing probability.

\paragraph{Solving the score system.}
By Lemma~\ref{score:centering}, the matrix $J$ is block lower
triangular with diagonal blocks $-I_e$ and is invertible. The ideal
Jacobian satisfies the same uniform law as before.
Its generated-data perturbation is bounded by
\eqref{dyn:criterion-bound}, using the mixed derivative envelope. Its
parameter perturbation is bounded by
$\|\widehat\theta-\theta_0\|N_T^{-1}\sum W(Z_t)=o_p(1)$.
Thus the integrated Jacobian
\[
 \overline J_T=\frac1{N_T}\sum_{t>b_T}\int_0^1
 \partial_\theta g\{\theta_0+u(\widehat\theta-\theta_0),
                    \widehat Z_t\}\,du
\]
satisfies, on the common localization event,
\begin{align*}
 \|\overline J_T-J\|
 &\le\left\|N_T^{-1}\sum_{t>b_T}\partial_\theta g(\theta_0,Z_t)-J\right\|
 +C\|\widehat\theta-\theta_0\|\mathbb P_TW
 +C\mathbb P_T\{W(Z_t)\|a_t\|\}\\
 &=O_p(N_T^{-1/2})+o_p(1)O_p(1)
       +O_p\{\sqrt{\log T/T}\}=o_p(1).
\end{align*}
The first term uses iid finite second moments, and the last uses
\eqref{dyn:criterion-bound}.
Writing $\kappa_J=\|J^{-1}\|^{-1}$, on
$\{\|\overline J_T-J\|\le\kappa_J/2\}$ the Neumann series and inverse
identity give
\begin{align}
 \|\overline J_T^{-1}\|&\le2/\kappa_J,&
 \|\overline J_T^{-1}-J^{-1}\|
 &\le2\kappa_J^{-2}\|\overline J_T-J\|=o_p(1).
 \label{dyn:jacobian-inverse}
\end{align}
The event has probability tending to one. The exact stacked equation is
\[
 0=N_T^{-1}\sum_{t>b_T}g(\theta_0,\widehat Z_t)
                              +\overline J_T(\widehat\theta-\theta_0).
\]
The leading term $L_T$ in \eqref{dyn:score-expansion} is the
weighted innovation-score sum in Lemma~\ref{score:array-limit}.
That lemma proves $L_T\Rightarrow N(0,\Omega_{\rm seq})$ and
$L_T=O_p(1)$ using the centering, harmonic sums, and finite
third moments established above.
Put $c_T=T/N_T$, so $c_T-1=b_T/N_T=o(1)$. The exact equation gives
$\sqrt T(\widehat\theta-\theta_0)=-c_T\overline J_T^{-1}(L_T+r_T)$.
Consequently,
\begin{align}
 \|\sqrt T(\widehat\theta-\theta_0)+J^{-1}L_T\|
 &\le\{\|\overline J_T^{-1}-J^{-1}\|
        +|c_T-1|\|\overline J_T^{-1}\|\}\|L_T\|\notag\\
 &\quad+c_T\|\overline J_T^{-1}\|\|r_T\|\notag\\
 &=\{o_p(1)+o(1)O_p(1)\}O_p(1)+O_p(1)o_p(1)=o_p(1).
 \label{dyn:final-normalization}
\end{align}
This proves \eqref{dyn:influence}.

\paragraph{Fitted score moments.}
For each scalar entry $G$ of $\partial_\theta g$, $A_g$, $gg'$,
$g\varphi'$, and $\varphi\varphi'$ used in the upper-left block
of \eqref{main:joint-feasible}, the polynomial derivative
envelopes and \eqref{dyn:z-perturbation} give, on $\mathcal E_T$,
\begin{align}
 \frac{1}{N_T}\sum_{t>b_T}
 |G(\widehat\theta,\widehat Z_t)-G(\theta_0,Z_t)|
 &\le\|\widehat\theta-\theta_0\|\frac{1}{N_T}\sum_{t>b_T}W_3(Z_t)
       +\frac{1}{N_T}\sum_{t>b_T}W_3(Z_t)\|a_t\|\\
 &=o_p(1)+O_p\{\sqrt{\log T/T}\}=o_p(1).
\end{align}
For entries without parameter dependence, the first term is zero.
The ideal empirical moments obey iid laws of large numbers. Therefore
$\widehat J\to_pJ$, $\widehat B_g\to_pB_g$, and all score and marginal
cross moments converge to the quantities in
\eqref{main:sequential-covariance}. The inverse identity in
\eqref{dyn:jacobian-inverse} then proves covariance consistency.
The contrast statement follows by Slutsky's theorem when its limiting
variance is positive.
\end{proof}

\begin{corollary}[Uniform and terminal margin errors]\label{dyn:margin-rates}
Under Proposition~\ref{dyn:theorem}, the fitted conditional margins satisfy
\begin{align}
 \max_{b_T<t\le T}\max_{i\le n}\sup_y
 |\widehat F_{i,t|t-1}(y)-F^0_{i,t}(y)|
 &=O_p\{(\log T/b_T)^{1/2}\},\label{dyn:uniform-cdf}\\
 \max_{i\le n}\sup_y
 |\widehat F_{i,T+1|T}(y)-F^0_{i,T+1}(y)|
 &=O_p(T^{-1/2}).\label{dyn:terminal-cdf}
\end{align}
\end{corollary}

\begin{proof}
The maximal bounds \eqref{dyn:prefix-max}, the inverse bound
\eqref{dyn:inverse}, and the exact residual-variance identity imply
\begin{align}
 \sup_{m\ge b_T,\,m\le T}\|d_{i,m}\|&=O_p(b_T^{-1/2}),&
 \sup_{m\ge b_T,\,m\le T}|a^v_{i,m}|&=O_p(b_T^{-1/2}),\\
 \max_{b_T<t\le T}|a^\mu_{it}|&=O_p\{(\log T/b_T)^{1/2}\}.
\end{align}
For $|v|\le1/2$, the derivatives of
$\Phi\{(z-a)/\sqrt{1+v}\}$ with respect to $a,v$ have uniformly bounded
absolute values over $z\in\R$ and $|a|\le1$:
with $w=(z-a)/\sqrt{1+v}$, they are
\[
 \partial_a\Phi(w)=-\frac{\phi(w)}{\sqrt{1+v}},\qquad
 \partial_v\Phi(w)=-\frac{w\phi(w)}{2(1+v)}.
\]
Their suprema are finite because $1+v\ge1/2$,
$\sup_w\phi(w)<\infty$, and $\sup_w|w|\phi(w)<\infty$.
Integrating these derivatives along $(sa,sv)$, $0\le s\le1$, gives
\[
 \sup_y|\widehat F_{i,t|t-1}(y)-F^0_{i,t}(y)|
 \le C(|a^\mu_{it}|+|a^v_{i,t-1}|)
\]
on an event with probability tending to one. This proves
\eqref{dyn:uniform-cdf}. At the terminal estimation date,
\eqref{dyn:fixed-moments} and \eqref{dyn:inverse} give
$d_{i,T}=O_p(T^{-1/2})$ and $a^v_{i,T}=O_p(T^{-1/2})$.
Stationarity gives $x_{T+1}=O_p(1)$. Their product is
$a^\mu_{i,T+1}=O_p(T^{-1/2})$ without any independence assumption.
The same CDF bound proves \eqref{dyn:terminal-cdf}.
\end{proof}

\subsection{Terminal influences and the joint array}

The terminal OLS normal equations and minimized residual sum of squares
give the coefficient and scale influences. Predictability and the
intercept identity determine their covariance with the innovation score.
We retain mixed third moments, which can be nonzero under asymmetric
copula dependence despite normal coordinate margins.

\begin{lemma}[Terminal linearization and moment identities]
\label{dji:endpoint-lemma}
Under the conditions of Theorem~\ref{dji:joint},
the terminal estimation errors in \eqref{main:terminal-influences} satisfy
\begin{equation}
 \sqrt T\delta_T=\frac1{\sqrt T}\sum_{t=1}^Te_t+o_p(1).
 \label{dji:endpoint-expansion}
\end{equation}
The array formed by the summands of $U_T$ is a martingale-difference
array for $\mathcal F_t=\sigma(Z_s:s\le t)$. Its covariance converges
to \eqref{main:joint-covariance}.
\end{lemma}

\begin{proof}
Put $s_{i,T}=T^{-1}\sum_{t\le T}x_tZ_{it}$ and
$u_{i,T}=T^{-1}\sum_{t\le T}(Z_{it}^2-1)$.
Lemma~\ref{dyn:moments} proves $Q>0$ and
\begin{equation}
 \|\widehat Q_T-Q\|=O_p(T^{-1/2}),\qquad
 \|s_{i,T}\|+|u_{i,T}|=O_p(T^{-1/2}).
 \label{dji:terminal-inputs}
\end{equation}
For $\kappa=\lambda_{\min}(Q)>0$, on
$\mathcal A_T=\{\|\widehat Q_T-Q\|\le\kappa/2\}$,
\begin{align}
 \lambda_{\min}(\widehat Q_T)&\ge\kappa/2,&
 \|\widehat A_Q\|&\le2/\kappa,\notag\\
 \widehat A_Q-A_Q&=-\widehat A_Q(\widehat Q_T-Q)A_Q,&
 \|\widehat A_Q-A_Q\|&\le2\kappa^{-2}\|\widehat Q_T-Q\|.
 \label{dji:inverse}
\end{align}
By the second-moment bound in \eqref{dyn:fixed-moments},
$P(\mathcal A_T^c)\le C/(T\kappa^2)\to0$.
For large $T$, $\kappa/2>T^{-1/4}$, so the coefficient threshold
is inactive on this event. The exact normal equation and minimized
residual-sum-of-squares identity are
\begin{align}
 d_{i,T}&=\widehat A_Q s_{i,T},\notag\\
 \frac1{T\sigma_i^2}\sum_{t\le T}
       (y_{it}-x_t'\widehat\beta_{i,T})^2
 &=1+u_{i,T}-s_{i,T}'\widehat A_Q s_{i,T}.
 \label{dji:terminal-exact}
\end{align}
For large $T$, $T^{-1/4}<\min_i\sigma_i^2/2$. On $\mathcal A_T$,
a variance floor can bind only if
$|u_{i,T}|+(2/\kappa)\|s_{i,T}\|^2\ge1/2$ for some $i$. Hence
\begin{align*}
 P(\text{some variance floor binds})
 &\le P(\mathcal A_T^c)+\sum_i
   \left[P(|u_{i,T}|\ge1/4)+P(\|s_{i,T}\|^2\ge\kappa/8)\right]\\
 &\le\frac CT+16\sum_iE|u_{i,T}|^2
                   +\frac8\kappa\sum_iE\|s_{i,T}\|^2
 \le\frac{C'}T\longrightarrow0.
\end{align*}
On the common event where neither thresholds nor floors bind,
\begin{align}
 \sqrt T\|d_{i,T}-A_Qs_{i,T}\|
 &\le2\kappa^{-2}\sqrt T\|\widehat Q_T-Q\|\|s_{i,T}\|\notag\\
 &=O_p(\sqrt T\,T^{-1/2}T^{-1/2})=O_p(T^{-1/2}),\notag\\
 \sqrt T|a_{i,T}^v-u_{i,T}|
 &\le(2/\kappa)\sqrt T\|s_{i,T}\|^2
 =O_p(\sqrt T\,T^{-1})=O_p(T^{-1/2}).
 \label{dji:terminal-remainders}
\end{align}
The complement of the common event has probability tending to zero.
Stacking the finitely many remainders proves
\eqref{dji:endpoint-expansion}.

Write $g_t=g(\theta_0,Z_t)$ and
$q_{T,t}=\mathbf1\{t>b_T\}g_t+h_{T,t}B_g\varphi_t$.
The interior population maximum at each stage gives $Eg_t=0$.
$EZ_t=0$ and $Ev_t=0$ follow from the normal coordinate margins.
The vector $x_t$ is $\mathcal F_{t-1}$-measurable and $Z_t$ is
independent of that field. Thus
\begin{equation}
 E\left[\begin{pmatrix}q_{T,t}\\w_t\\v_t\end{pmatrix}
           \,\middle|\,\mathcal F_{t-1}\right]=0,
 \label{dji:martingale}
\end{equation}
Distinct dates are therefore orthogonal, although the coefficient
influences share lagged innovations. All required moments follow from
Lemma~\ref{dyn:moments} and the score envelopes.

For any column-valued function $f(Z_t)$ with the needed moments,
current-innovation independence and $A_Q\mu_x=e_1$ give
\begin{align}
 E\{f(Z_t)w_{i,t}'\}
 &=E\{f(Z_t)Z_{it}\}\mu_x'A_Q
 =E\{f(Z_t)Z_{it}\}e_1',\notag\\
 E(w_{i,t}w_{j,t}')
 &=E(Z_{it}Z_{jt})A_Q E(x_tx_t')A_Q
 =(R_Z)_{ij}A_Q,\notag\\
 E\{w_{i,t}(v_t)_j\}
 &=E\{Z_{it}(Z_{jt}^2-1)\}A_Q\mu_x
 =(M_3)_{ij}e_1.
 \label{dji:cross-identities}
\end{align}
The identity $A_Q\mu_x=e_1$ is \eqref{dyn:intercept-identity}.
It uses the actual intercept in the design. These equalities
prove $E(e_te_t')=W_\delta$ including every coefficient--variance
cross block, with no assumption of sign symmetry of the copula.

Let $A_{ge}=E(g_te_t')$ and $A_{\varphi e}=E(\varphi_te_t')$.
Orthogonality at distinct dates gives the exact finite-sample block
\begin{align}
 \operatorname{Cov}\left(T^{-1/2}\sum_tq_{T,t},
                         T^{-1/2}\sum_te_t\right)
 &=\frac{N_T}{T}A_{ge}
      +\frac1T\sum_th_{T,t}B_g A_{\varphi e}\notag\\
 &=\frac{N_T}{T}(A_{ge}+B_gA_{\varphi e}).
 \label{dji:cross-finite}
\end{align}
Indeed, reversing the finite sum gives
$\sum_th_{T,t}=\sum_{m=b_T}^{T-1}m/m=N_T$.
Applying the first equality in \eqref{dji:cross-identities} to
$f=g+B_g\varphi$ proves
$A_{ge}+B_gA_{\varphi e}=(L_\beta,M_v)$.
Multiplication by $-J^{-1}$ gives the sign and ordering of
$D_{\rm VAR}$ in \eqref{main:joint-covariance}.
The score block has exact covariance
\begin{equation}
 \frac{N_T}{T}S+
 \left(\frac1T\sum_{t>b_T}h_{T,t}\right)
       (\Gamma B_g'+B_g\Gamma')+
 \left(\frac1T\sum_th_{T,t}^2\right)B_g\Sigma_\varphi B_g'.
 \label{dji:score-covariance}
\end{equation}
The harmonic identities \eqref{eq:prefix-harmonic-limits} and
\eqref{eq:sequential-cross-weight} give coefficient limits $1,1,2$.
Their burn-in error tends to zero since
$(b_T/T)\log(T/b_T)\to0$. The lower block equals $W_\delta$
exactly. This proves the asserted covariance limit.
\end{proof}

\begin{lemma}[Separated-block proof for the joint influence array]
\label{dji:block-clt}
Under the conditions of Theorem~\ref{dji:joint}, the array $U_T$ in
\eqref{dji:joint-expansion} converges to
$N(0,\mathcal W_{\rm VAR})$.
\end{lemma}

\begin{proof}
The coefficient influences share lagged innovations across dates.
Truncating the stable filter lets us separate them into independent
blocks while bounding the discarded part of the array.
Fix any deterministic vector $a$ of the joint dimension. Write
$a'U_T=T^{-1/2}\sum_{t=1}^TY_{T,t}$.
By \eqref{dyn:linear-design}, for each integer $q\ge1$ define
\begin{equation}
 x_t^{(q)}=\mu_x+\sum_{j=1}^qC_j\varepsilon_{t-j}.
 \label{dji:truncated-design}
\end{equation}
For every fixed $r\ge1$, Minkowski's inequality gives
\begin{align}
 \|x_t-x_t^{(q)}\|_{L^r}
 &\le\sum_{j>q}\|C_j\|\|\varepsilon_0\|_{L^r}
 \le C_r\sum_{j>q}\rho^j\le C_r\rho^q,\notag\\
 \sup_{q,t}\|x_t^{(q)}\|_{L^r}
 &\le\|\mu_x\|+\sum_{j\ge1}C\rho^j\|\varepsilon_0\|_{L^r}
 <\infty.
 \label{dji:lag-approximation}
\end{align}
Replace only $x_t$ in $Y_{T,t}$ by $x_t^{(q)}$, obtaining
$Y_{T,t}^{(q)}$. This new variable is a function of
$(Z_{t-q},\ldots,Z_t)$. It is centered conditionally on
$\mathcal F_{t-1}$, because $x_t^{(q)}$ is predictable and
$EZ_t=0$. The difference is also a martingale difference.
Partition $a$ conformably into $(a_\theta,(a_{\beta_i})_{i\le n},a_v)$.
Only its coefficient terms change, so
\begin{align*}
 Y_{T,t}-Y_{T,t}^{(q)}
 &=\sum_i a_{\beta_i}'Q^{-1}(x_t-x_t^{(q)})Z_{it},\\
 E|Y_{T,t}-Y_{T,t}^{(q)}|^2
 &\le n\sum_i\|a_{\beta_i}'Q^{-1}\|^2
                E\|x_t-x_t^{(q)}\|^2 EZ_{it}^2
 \le C_a\rho^{2q}.
\end{align*}
The factorization uses independence of the current innovation from
both designs. Martingale orthogonality therefore yields
\begin{align}
 E\left|T^{-1/2}\sum_t(Y_{T,t}-Y_{T,t}^{(q)})\right|^2
 &=T^{-1}\sum_tE|Y_{T,t}-Y_{T,t}^{(q)}|^2
 \le C_a\rho^{2q}.
 \label{dji:sum-approximation}
\end{align}
The coefficient weights have not been enlarged by $h_{T,t}$:
that weight multiplies only the unchanged innovation function.

Set $q=q_T=\lceil c\log T\rceil$ for any fixed $c>0$, and
$\ell=\ell_T=\lfloor T^{1/8}\rfloor$. Starting at date one,
retain consecutive blocks $I_j$ of length $\ell$, separate them
by $q$ omitted dates, and omit the final incomplete block. Let
$D_T$ denote all omitted dates, and $m_T$ the number of retained
blocks. Their deterministic sizes satisfy
\begin{equation}
 m_T\le T/\ell,\qquad
 |D_T|\le Tq/\ell+\ell+q.
 \label{dji:block-count}
\end{equation}
The retained block sums depend on disjoint innovation sets:
a block $[r,s]$ uses only dates $[r-q,s]$, and the next retained
block starts at $s+q+1$. Thus these block sums are independent,
even for the block whose innovation support extends to dates
zero or below in the stationary initialization.

The polynomial score envelopes, \eqref{dji:lag-approximation},
and independence of current innovations from the design imply,
for $r=2,3$,
\begin{equation}
 \sup_t E|Y_{T,t}^{(q)}|^r\le C_a(1+\log T)^r.
 \label{dji:array-moments}
\end{equation}
Here $0\le h_{T,t}\le1+\log T$, and the remaining coefficients
are fixed. Orthogonality of the truncated martingale differences
also holds on any deterministic subset. Consequently
\begin{align}
 E\left|T^{-1/2}\sum_{t\in D_T}Y_{T,t}^{(q)}\right|^2
 &\le\frac{C_a|D_T|}{T}(1+\log T)^2\notag\\
 &\le C_a\left(\frac q\ell+\frac{\ell+q}{T}\right)
                         (1+\log T)^2\notag\\
 &=O\{T^{-1/8}(1+\log T)^3+
                  T^{-7/8}(1+\log T)^2\}=o(1).
 \label{dji:gap-bound}
\end{align}
Let $B_{T,j}=T^{-1/2}\sum_{t\in I_j}Y_{T,t}^{(q)}$ and
$S_T^B=\sum_jB_{T,j}$. The $B_{T,j}$ are independent and centered.
For arbitrary real numbers $z_t$, H\"older's inequality gives
$|\sum_{t\in I_j}z_t|^3\le\ell^2\sum_{t\in I_j}|z_t|^3$.
It follows that
\begin{align}
 \sum_jE|B_{T,j}|^3
 &\le\frac{m_T\ell^3C_a(1+\log T)^3}{T^{3/2}}
 \le\frac{C_a\ell^2(1+\log T)^3}{\sqrt T}\notag\\
 &=O\{T^{-1/4}(1+\log T)^3\}=o(1),\notag\\
 \max_j E B_{T,j}^2
 &\le C_a\ell(1+\log T)^2/T
 =O\{T^{-7/8}(1+\log T)^2\}=o(1).
 \label{dji:block-lyapunov}
\end{align}
The variance bound uses martingale orthogonality within each block.
By \eqref{dji:sum-approximation} and \eqref{dji:gap-bound},
$\|a'U_T-S_T^B\|_{L^2}\to0$. Lemma~\ref{dji:endpoint-lemma}
gives $E(a'U_T)^2\to s_a^2=a'\mathcal W_{\rm VAR}a$ and
bounded $L^2$ norms. The triangle inequality bounds
$\|S_T^B\|_{L^2}$. Cauchy--Schwarz then gives
\begin{align}
 |E(S_T^B)^2-E(a'U_T)^2|
 &\le\|S_T^B-a'U_T\|_{L^2}
       (\|S_T^B\|_{L^2}+\|a'U_T\|_{L^2})\to0.
 \label{dji:block-variance-limit}
\end{align}
Hence $\sum_jE B_{T,j}^2\to s_a^2$.

The retained blocks are independent and centered. For every
$\epsilon>0$, \eqref{dji:block-lyapunov} gives
\begin{equation}
 \sum_j E\{B_{T,j}^2\mathbf1(|B_{T,j}|>\epsilon)\}
 \le\epsilon^{-1}\sum_j E|B_{T,j}|^3=o(1).
 \label{dji:block-lindeberg}
\end{equation}
Together with \eqref{dji:block-variance-limit}, this verifies
\citet[Theorem~3.4.10]{durrett_2019} for $S_T^B$ when $s_a^2>0$.
If $s_a^2=0$, $S_T^B\to0$ in $L^2$. The approximation
$\|a'U_T-S_T^B\|_{L^2}\to0$ and the Cram\'er--Wold criterion
complete the vector limit, including singular covariances.
\end{proof}

\subsection{Completion of the joint inference theorem}

\begin{proof}[Proof of Theorem~\ref{dji:joint}]
The copula expansion is \eqref{dyn:influence}, whose generated-score
and predictable-design remainder bounds are proved in
Section~\ref{dyn:section}. Combining it with
\eqref{dji:endpoint-expansion} gives $\sqrt T\Delta_T=U_T+o_p(1)$ with
$U_T$ defined in \eqref{dji:joint-expansion}. No independence between the remainder
vectors is required: their dimensions are fixed and each is
$o_p(1)$. Lemma~\ref{dji:block-clt} supplies the joint limit of
the leading array, and addition of the vanishing vector proves
the asserted joint normal limit.

For the covariance estimator, define
\[
 \mathcal C=
 \begin{pmatrix}J^{-1}B_g\Sigma_\varphi B_g'J^{-\prime}&0\\0&0
 \end{pmatrix}.
\]
Direct expansion using \eqref{dji:ideal-xi} and
\eqref{dji:cross-identities} gives
\begin{align}
 E(\xi_t\xi_t')+\mathcal C
 &=\begin{pmatrix}
 J^{-1}\{S+\Gamma B_g'+B_g\Gamma'
                   +2B_g\Sigma_\varphi B_g'\}J^{-\prime}
 &-J^{-1}(L_\beta,M_v)\\
 -(L_\beta,M_v)'J^{-\prime}&W_\delta
 \end{pmatrix}\notag\\
 &=\mathcal W_{\rm VAR}.
 \label{dji:psd-identity}
\end{align}
The harmonic correction adds a second copy of the nuisance variance to
the copula block. The cross blocks already have coefficient one and
need no further correction.
Lemma~\ref{dji:moment-lemma} proves convergence of that outer
product. The innovation-moment consistency in
Proposition~\ref{dyn:theorem} and the inverse bound
\eqref{dji:J-inverse} prove convergence of the additional block
to $\mathcal C$. For the product calculation put
$K=J^{-1}B_g$, $\widehat K=\widehat J^{-1}\widehat B_g$,
$S_\varphi=\Sigma_\varphi$, and
$\widehat S_\varphi=\widehat\Sigma_\varphi$. Then
\begin{align*}
 \widehat K-K
 &=(\widehat J^{-1}-J^{-1})\widehat B_g
                       +J^{-1}(\widehat B_g-B_g),\\
 \widehat K\widehat S_\varphi\widehat K'-KS_\varphi K'
 &=(\widehat K-K)\widehat S_\varphi\widehat K'
       +K(\widehat S_\varphi-S_\varphi)\widehat K'
       +KS_\varphi(\widehat K'-K').
\end{align*}
On the common bounded-matrix event, these identities imply
\begin{align*}
 \|\widehat K-K\|
 &\le C(\|\widehat J^{-1}-J^{-1}\|+\|\widehat B_g-B_g\|)=o_p(1),\\
 \|\widehat{\mathcal W}_{\rm VAR}-\mathcal W_{\rm VAR}\|
 &\le\left\|N_T^{-1}\sum_{t>b_T}\widehat\xi_t\widehat\xi_t'
                              -E\xi_t\xi_t'\right\|
       +C(\|\widehat K-K\|+\|\widehat S_\varphi-S_\varphi\|)
 =o_p(1).
\end{align*} Both terms of \eqref{main:joint-feasible} are positive semidefinite:
the first is an average of outer products and the second is
a congruence of the empirical outer-product matrix
$\widehat\Sigma_\varphi$. The identity-matrix convention is positive
definite and is used with probability tending to zero, by
\eqref{dji:inverse} and \eqref{dji:J-inverse}. This proves
consistency and positive semidefiniteness for every sample.
Corollary~\ref{dyn:margin-rates} supplies the terminal-margin conclusion.
\end{proof}

\begin{remark}[Concrete family scope]
The Gaussian fixed-vine verification in
Proposition~\ref{score:gaussian-family} and the positive-parameter
Clayton verification in Proposition~\ref{tail:sequential-clayton}
verify the same conditions for sequential estimation used here. Consequently the
joint theorem applies to both families under the stated stable
VAR assumptions, including estimated higher-tree ancestors in
the Clayton case. It does not impose Gaussian joint innovations:
normal coordinate margins supply the moments, while $R_Z,M_3,S_v$
and the score cross moments retain the actual innovation copula.
\end{remark}

\section{Inference for the Conditional Event Probability at the Terminal VAR State}
\label{dfe:section}

The terminal lag state is random and belongs to the estimation sample,
so event inference requires more than an estimator CLT. We first
separate the state from the leading estimation error. We then expand
the event map on compact state sets and use stationarity to remove
the compact restriction. Localization also permits studentization
under almost-sure positivity of the projected variance. All
probabilities refer to the stationary law in Section~\ref{sec:method}.

\subsection{Separation from the terminal state}

Stability permits approximation of the terminal state by recent
innovations. Deleting those dates from the normalized influence sum
has a negligible effect. The two approximations depend on disjoint
innovation sets, which yields the independent joint limit required
by the forecast delta method.

Use $\Delta_T$ from \eqref{main:terminal-influences}, put
$d_\Delta=q_\theta+nd_x+n$, and write the joint influence array as
\begin{equation}
 q_{T,s}=\mathbf1\{s>b_T\}g(\theta_0,Z_s)+h_{T,s}B_g\varphi_s,
 \qquad L_{T,s}=(-q_{T,s}'J^{-\prime},e_s')',
 \qquad U_T=T^{-1/2}\sum_{s=1}^T L_{T,s}.
 \label{dfe:joint-array}
\end{equation}
Let $X\stackrel d=x_1$ and take $G\sim N(0,\mathcal W_{\rm VAR})$
independent of $X$.

\begin{lemma}[Separation from the terminal state]\label{dfe:state-separation}
Under the conditions of Theorem~\ref{dji:joint},
\begin{equation}
 (\sqrt T\Delta_T,x_{T+1})\ \Longrightarrow\ (G,X).
 \label{dfe:state-limit}
\end{equation}
The conclusion allows singular $\mathcal W_{\rm VAR}$.
\end{lemma}

\begin{proof}
By \eqref{dji:martingale}, $(L_{T,s})$ is a martingale-difference
array for $\mathcal F_s$.
The finite score and
innovation moments and Lemma~\ref{dyn:moments} imply
\begin{align}
 E\|q_{T,s}\|^2
 &\le2E\|g(\theta_0,Z_s)\|^2
          +2h_{T,s}^2\|B_g\|^2E\|\varphi(Z_s)\|^2
 \le C(1+h_{T,s}^2),\notag\\
 E\|Z_{is}Q^{-1}x_s\|^2
 &\le\|Q^{-1}\|^2E\|x_s\|^2EZ_{is}^2\le C,&
 E(Z_{is}^2-1)^2&=2,\notag\\
 E\|L_{T,s}\|^2&\le C(1+h_{T,s}^2).
 \label{dfe:increment-moments}
\end{align}
The factorization in the second line uses temporal independence, not
independence between coordinates. If $s<t$, conditional expectation at
$\mathcal F_{t-1}$ gives $E(L_{T,s}'L_{T,t})=0$.

Set $k_T=\lceil\log T\rceil$ for $T\ge3$ and define
\begin{align}
 U_T^-&=T^{-1/2}\sum_{s=1}^{T-k_T}L_{T,s},&
 X_T^+&=\mu_x+\sum_{j=1}^{k_T}C_j\varepsilon_{T+1-j}.
 \label{dfe:cut-objects}
\end{align}
For all sufficiently large $T$, $b_T<T-k_T$ and $k_T<T/2$.
For $T-k_T<s\le T$, the harmonic-weight definition gives
\begin{equation}
 0\le h_{T,s}=\sum_{m=s}^{T-1}m^{-1}
 \le\frac{T-s}{s}\le\frac{k_T}{T-k_T}\le\frac{2k_T}{T},
 \label{dfe:last-weights}
\end{equation}
where the sum is zero for $s=T$. Orthogonality and
\eqref{dfe:increment-moments} yield
\begin{align}
 E\|U_T-U_T^-\|^2
 &=T^{-1}\sum_{s=T-k_T+1}^TE\|L_{T,s}\|^2\notag\\
 &\le C\frac{k_T}{T}
           \left(1+\frac{4k_T^2}{T^2}\right)
 =O(\log T/T)\longrightarrow0.
 \label{dfe:deleted-block}
\end{align}
The causal filter \eqref{dyn:linear-design} and Minkowski's inequality
supply the other approximation:
\begin{align}
 \|x_{T+1}-X_T^+\|_{L^2}
 &\le\sum_{j>k_T}\|C_j\|\,\|\varepsilon_1\|_{L^2}
 \le C\sum_{j>k_T}\rho^j
 \le C\rho^{k_T}\longrightarrow0.
 \label{dfe:state-truncation}
\end{align}
Here $U_T^-$ is $\mathcal F_{T-k_T}$-measurable, whereas $X_T^+$
depends only on innovations with indices $T-k_T+1,\ldots,T$.
The iid innovation assumption makes these two objects independent.

For fixed $a\in\mathbb R^{d_\Delta}$ and $c\in\mathbb R^{d_x}$,
$|e^{iu}-e^{iv}|\le|u-v|$ and the two preceding bounds give
\begin{align}
 &\left|E e^{i(a'U_T+c'x_{T+1})}
                 -E e^{ia'U_T^-}E e^{ic'X_T^+}\right|\notag\\
 &\qquad\le\|a\|E\|U_T-U_T^-\|
             +\|c\|E\|x_{T+1}-X_T^+\|\notag\\
 &\qquad\le C\|a\|\sqrt{\log T/T}
             +C\|c\|\rho^{k_T}\longrightarrow0.
 \label{dfe:characteristic-separation}
\end{align}
By \eqref{dfe:deleted-block} and the joint estimator CLT,
$U_T^-\Longrightarrow G$. By stationarity, $x_{T+1}$ has the law of
$X$ for every $T$, so \eqref{dfe:state-truncation} implies
$X_T^+\Longrightarrow X$. The product on the first line of
\eqref{dfe:characteristic-separation} therefore converges to
$E e^{ia'G}E e^{ic'X}$, the characteristic function of an independent
pair. The characteristic-function convergence theorem proves
$(U_T,x_{T+1})\Longrightarrow(G,X)$.
Finally $\sqrt T\Delta_T-U_T=o_p(1)$ from
Theorem~\ref{dji:joint}. Adding this vanishing vector proves
\eqref{dfe:state-limit}.
\end{proof}

\subsection{Proof of Theorem~\ref{main:event-inference}}

Use $z_i(r),u_i(r)$ and $d(r)$ from \eqref{main:event-gradient},
$p(r)=C_{\theta_0}\{u(r)\}$, and
$p_{0,T+1}=p(x_{T+1})$. The fitted probability is
$\widehat p_{T+1}=\widehat p_{T+1}(c)$. Its estimated variance is
\begin{equation}
 \widehat\tau_T^2=\widehat d_T(x_{T+1})'
       \widehat{\mathcal W}_{\rm VAR}\widehat d_T(x_{T+1}).
 \label{dfe:feasible-variance}
\end{equation}

\begin{proof}[Proof of Theorem~\ref{main:event-inference}]
For coefficient perturbations $b_i\in\mathbb R^{d_x}$ and scale
perturbations $a_i>-1$, stack
\[
 \delta=(b_1',\ldots,b_n',a_1,\ldots,a_n)'.
\]
Define
\begin{equation}
 H(r,\theta,\delta)=C_\theta\left(
     \Phi\left\{\frac{z_i(r)-r'b_i}{\sqrt{1+a_i}}\right\}:i\le n\right).
 \label{dfe:exact-map}
\end{equation}
The normalization of $d_{i,T}$ and $a^v_{i,T}$ implies the exact identity
\[
 \widehat z_{i,T}(r)
       =\{z_i(r)-r'd_{i,T}\}/\sqrt{1+a^v_{i,T}}.
\]
Consequently $H(x_{T+1},\theta_0,0)=p_{0,T+1}$ and
$H(x_{T+1},\widehat\theta,\delta_T)=\widehat p_{T+1}$.
Differentiation at $(\theta_0,0)$ gives
\[
 \partial_{b_i}\{(z_i-r'b_i)(1+a_i)^{-1/2}\}=-r,
 \qquad
 \partial_{a_i}\{(z_i-r'b_i)(1+a_i)^{-1/2}\}=-z_i/2,
\]
which proves every sign and normalization in
\eqref{main:event-gradient} by the chain rule.

The random state need not lie in a fixed compact set. We bound the
Hessian on each such set first, then let its radius increase.
Fix $M<\infty$. For $\|r\|\le M$,
\begin{equation}
 |z_i(r)|\le K_M:={\max_{i\le n}
               (|c_i|+M\|\beta_i\|)/\sigma_i}<\infty.
 \label{dfe:compact-arguments}
\end{equation}
Choose $0<\epsilon_M\le1/2$ so the closed parameter ball of radius
$\epsilon_M$ around $\theta_0$ is contained in
the neighborhood on which the event CDF is $C^2$. When $\|\delta\|\le\epsilon_M$,
\begin{align}
 1+a_i&\ge1/2,&
 \left|\frac{z_i(r)-r'b_i}{\sqrt{1+a_i}}\right|
 &\le\sqrt2(K_M+M\epsilon_M)=:K_M^*.
 \label{dfe:compact-map}
\end{align}
All copula arguments then belong to the compact cube
$[\Phi(-K_M^*),\Phi(K_M^*)]^n\subset(0,1)^n$.
Assumption~\ref{main:event-regularity}(i) and the differentiable normal and scale maps imply
that the Hessian of $H$ with respect to $(\theta,\delta)$ is continuous
on the corresponding compact set. Its operator norm has a finite
maximum $A_M$. The gradient $d(r)$ also has a finite maximum norm
$D_M$ on $\|r\|\le M$.

Define the remainder globally by
$R_T=\widehat p_{T+1}-p_{0,T+1}-d(x_{T+1})'\Delta_T$.
On the event $\{\|x_{T+1}\|\le M,\|\Delta_T\|\le\epsilon_M\}$,
Taylor's formula with integral remainder gives
\begin{align}
 R_T&=\int_0^1(1-s)\Delta_T'
   \partial^2_{(\theta,\delta)}H
    (x_{T+1},\theta_0+s(\widehat\theta-\theta_0),s\delta_T)
                                     \Delta_T\,ds,\notag\\
 \sqrt T|R_T|&\le\tfrac12 A_M\sqrt T\|\Delta_T\|^2
                 =O_p(T^{-1/2}).
 \label{dfe:taylor-bound}
\end{align}
The last order uses $\sqrt T\Delta_T=O_p(1)$ from
Theorem~\ref{dji:joint}, with $M$ fixed. For every $\varepsilon>0$,
\begin{align}
 P(\sqrt T|R_T|>\varepsilon)
 &\le P(\|X\|>M)+P(\|\Delta_T\|>\epsilon_M)
       +P(A_M\sqrt T\|\Delta_T\|^2/2>\varepsilon),\notag\\
 \limsup_{T\to\infty}P(\sqrt T|R_T|>\varepsilon)
 &\le P(\|X\|>M)\longrightarrow0\quad(M\to\infty).
 \label{dfe:unlocalize-taylor}
\end{align}
Hence
\begin{equation}
 \sqrt T(\widehat p_{T+1}-p_{0,T+1})
 =d(x_{T+1})'\sqrt T\Delta_T+o_p(1).
 \label{dfe:event-expansion}
\end{equation}
Continuity of $d$,
Lemma~\ref{dfe:state-separation}, and the continuous mapping theorem
give \eqref{main:event-limit}. Independence of $G$ and $X$ identifies
its conditional Gaussian law.

The same localization controls the fitted gradient at the terminal
state. For $|a|\le1/2$, the derivative bound
$|\partial_a(1+a)^{-1/2}|\le\sqrt2$ yields
$|(1+a)^{-1/2}-1|\le\sqrt2|a|$.
Using the exact normal-argument identity, on
$\|r\|\le M$ and $\|\delta_T\|\le1/2$ we obtain
\begin{align}
 |\widehat z_{i,T}(r)-z_i(r)|
 &\le |z_i(r)|\,|(1+a^v_{i,T})^{-1/2}-1|
       +|r'd_{i,T}|(1+a^v_{i,T})^{-1/2}\notag\\
 &\le\sqrt2 K_M|a^v_{i,T}|+\sqrt2 M\|d_{i,T}\|
 \le C_M\|\delta_T\|.
 \label{dfe:plug-arguments}
\end{align}
The first derivatives of $C$ are locally Lipschitz on this compact
set because its second derivatives are continuous there. $\phi$, its
derivative, and multiplication by $z_i(r)$ and $r$ have bounded
derivatives on the same set. The mean-value inequality therefore gives
\begin{equation}
 \sup_{\|r\|\le M}\|\widehat d_T(r)-d(r)\|
       \le B_M\|\Delta_T\|=O_p(T^{-1/2})
 \label{dfe:plug-gradient}
\end{equation}
on a parameter event with probability tending to one, for a finite
constant $B_M$. Applying the probability decomposition in
\eqref{dfe:unlocalize-taylor}, now to this gradient difference, proves
\[
 \widehat d_T(x_{T+1})-d(x_{T+1})=o_p(1),\qquad
 d(x_{T+1})=O_p(1),\qquad \widehat d_T(x_{T+1})=O_p(1).
\]
The second assertion follows directly by stationarity and continuity
of $d$, and the third follows from the first two.
For $\widehat d=\widehat d_T(x_{T+1})$ and $d=d(x_{T+1})$, expand
and bound the variance difference as
\begin{align}
 |\widehat\tau_T^2-\tau^2(x_{T+1})|
 &\le\|\widehat{\mathcal W}_{\rm VAR}-\mathcal W_{\rm VAR}\|
                                      \|\widehat d\|^2\notag\\
 &\quad+\|\mathcal W_{\rm VAR}\|\,\|\widehat d-d\|
                                      (\|\widehat d\|+\|d\|)\notag\\
 &=o_p(1)O_p(1)+O(1)o_p(1)O_p(1)=o_p(1).
 \label{dfe:variance-consistency}
\end{align}
The covariance consistency here is precisely the feasible conclusion
of Theorem~\ref{dji:joint}.

Let $V_T=\tau^2(x_{T+1})$ and $V=\tau^2(X)$. Stationarity gives
$V_T\stackrel{d}=V$ for every $T$. If $P(V>0)=1$, then
$P(V\le\gamma)\downarrow0$ as $\gamma\downarrow0$. In particular,
for every $\gamma>0$,
\begin{align}
 P(\widehat\tau_T^2=0)
 &\le P(V\le\gamma)
         +P(|\widehat\tau_T^2-V_T|\ge\gamma),\notag\\
 P\left(\left|\frac{\widehat\tau_T^2-V_T}{V_T}\right|>\varepsilon\right)
 &\le P(V\le\gamma)
         +P(|\widehat\tau_T^2-V_T|>\varepsilon\gamma).
 \label{dfe:small-variance}
\end{align}
Taking first the limit superior in $T$, then $\gamma\downarrow0$,
proves that the zero-variance convention is asymptotically irrelevant
and that the relative variance error is $o_p(1)$.
Equations~\eqref{main:event-limit} and~\eqref{dfe:variance-consistency}
therefore give the joint limit
\[
 \left(\frac{\sqrt T(\widehat p_{T+1}-p_{0,T+1})}
                  {\widehat\tau_T},x_{T+1}\right)
 \Longrightarrow
 \left(\frac{d(X)'G}{\sqrt{d(X)'\mathcal W_{\rm VAR}d(X)}},X\right).
\]
For every state outside a set of $X$-probability zero, the first
coordinate on the right is a centered Gaussian divided by its positive
standard deviation. Its conditional distribution is $N(0,1)$ and
does not depend on $X$, hence its marginal distribution is $N(0,1)$.
This proves \eqref{main:pivot}. Continuity of the standard normal
CDF at $\pm z_{1-\alpha/2}$ proves \eqref{main:interval}. Intersecting
with $[0,1]$ preserves membership of $p_{0,T+1}\in[0,1]$.
\end{proof}

\subsection{Fixed rectangles and mixed directions}

For a fixed rectangle $A$ as in Lemma~\ref{dfe:rectangle-smooth}, define
\begin{align}
 p_{A,0,T+1}&=P(y_{T+1}\in A\mid\mathcal F_T)
              =\int_{\mathcal B(x_{T+1})}f_{\theta_0}(z)\,dz,
                 \notag\\
 \widehat p_{A,T+1}&=
     \int_{\widehat{\mathcal B}_T(x_{T+1})}f_{\widehat\theta}(z)\,dz,
 &\widehat z_{i,T}^- (r)&=
     \frac{\ell_i-r'\widehat\beta_{i,T}}{\sqrt{\widehat v_{i,T}}},\notag\\
 \widehat z_{i,T}^+ (r)&=
     \frac{u_i-r'\widehat\beta_{i,T}}{\sqrt{\widehat v_{i,T}}},&
 \widehat{\mathcal B}_T(r)&=
     \prod_i(\widehat z_{i,T}^-(r),\widehat z_{i,T}^+(r)].
 \label{dfe:rectangle-estimator}
\end{align}
Infinite endpoints remain unchanged. Set
$\tau_A^2(r)=d_A(r)'\mathcal W_{\rm VAR}d_A(r)$.
The fitted gradient substitutes $\widehat\theta$ and the finite
fitted endpoints in \eqref{dfe:rectangle-gradient}.

\begin{corollary}[Inference for a fixed rectangle]\label{dfe:rectangle-clt}
Under Theorem~\ref{dji:joint} and \eqref{main:event-envelope},
the expansion, joint mixed normal limit, and variance consistency of
Theorem~\ref{dfe:event-clt} hold for the rectangle estimates above.
If $P\{\tau_A^2(X)>0\}=1$, its studentized limit and coverage
conclusions also hold.
\end{corollary}

\begin{proof}
Let $H_A(r,\theta,\delta)$ be the probability map in
Lemma~\ref{dfe:rectangle-smooth}. Each finite fitted endpoint satisfies
exactly
\[
 \widehat z_{i,T}^\pm(r)
       =\{z_i^\pm(r)-r'd_{i,T}\}/\sqrt{1+a^v_{i,T}}.
\]
The target and estimator therefore equal
$H_A(x_{T+1},\theta_0,0)$ and
$H_A(x_{T+1},\widehat\theta,\delta_T)$, respectively. For each fixed
$M$, Lemma~\ref{dfe:rectangle-smooth} supplies a radius
$\epsilon_{A,M}>0$, a Hessian bound $A_{A,M}<\infty$, and a gradient
bound $D_{A,M}<\infty$ on $\|r\|\le M$ and the corresponding
parameter ball. Its derivatives are continuous and its second
derivatives bounded there. Repeating the Taylor identity itself for
this map gives, on that event,
\begin{align}
 \widehat p_{A,T+1}-p_{A,0,T+1}
   &=d_A(x_{T+1})'\Delta_T+R_{A,T},\notag\\
 \sqrt T|R_{A,T}|&\le A_{A,M}\sqrt T\|\Delta_T\|^2/2
                         =O_p(T^{-1/2}),\notag\\
 P(\sqrt T|R_{A,T}|>\varepsilon)
   &\le P(\|X\|>M)+P(\|\Delta_T\|>\epsilon_{A,M})
           +P(A_{A,M}\sqrt T\|\Delta_T\|^2/2>\varepsilon).
 \label{dfe:rectangle-remainder}
\end{align}
The last two probabilities tend to zero for fixed $M$. Then
$P(\|X\|>M)\to0$ as $M\to\infty$.
The bound \eqref{dfe:plug-arguments} applies to every finite endpoint,
with $K_M$ replaced by the maximum absolute value of these finitely many
endpoints. The bounded Hessian of the endpoint-probability map then
gives a finite $B_{A,M}$ such that
\[
 \sup_{\|r\|\le M}\|\widehat d_{A,T}(r)-d_A(r)\|
       \le B_{A,M}\|\Delta_T\|=O_p(T^{-1/2})
\]
on the same type of parameter event. Inactive endpoints contribute no
terms. The probability decomposition in
\eqref{dfe:rectangle-remainder} consequently gives
$\widehat d_{A,T}(x_{T+1})-d_A(x_{T+1})=o_p(1)$.
Put $\widehat\tau_{A,T}^2=
\widehat d_{A,T}(x_{T+1})'\widehat{\mathcal W}_{\rm VAR}
\widehat d_{A,T}(x_{T+1})$ and $V_T=\tau_A^2(x_{T+1})$.
The gradient bound, tightness of $d_A(x_{T+1})$, and covariance
consistency give
\begin{align*}
 |\widehat\tau_{A,T}^2-V_T|
 &\le\|\widehat{\mathcal W}_{\rm VAR}\|
       \|\widehat d_{A,T}-d_A\|(\|\widehat d_{A,T}\|+\|d_A\|)
       +\|d_A\|^2\|\widehat{\mathcal W}_{\rm VAR}-\mathcal W_{\rm VAR}\|
 =o_p(1).
\end{align*}
All gradients in this inequality are evaluated at $x_{T+1}$.
Lemma~\ref{dfe:state-separation} and the Taylor expansion now imply
\[
 (\sqrt T(\widehat p_{A,T+1}-p_{A,0,T+1}),
                      \widehat\tau_{A,T}^2,x_{T+1})
 \Longrightarrow(d_A(X)'G,\tau_A^2(X),X).
\]
For $c,\epsilon>0$, stationarity gives
\[
 P\{|\widehat\tau_{A,T}^2/V_T-1|>\epsilon\}
 \le P\{\tau_A^2(X)\le c\}
       +P\{|\widehat\tau_{A,T}^2-V_T|>\epsilon c\}.
\]
Taking $T\to\infty$ and then $c\downarrow0$ proves relative variance
consistency and excludes zero denominators with probability tending
to one. Conditional on $X$,
$d_A(X)'G/\tau_A(X)\sim N(0,1)$, so
\[
 P\{d_A(X)'G/\tau_A(X)\le z\}
 =E\{\Phi(z)\}=\Phi(z).
\]
The continuous mapping theorem gives the normal pivot. Its continuous
limit gives coverage $1-\alpha$, and intersecting with $[0,1]$
preserves coverage of the true probability.
\end{proof}

\subsection{Nondegeneracy for one-sided events}

For a one-sided event, each active marginal mean changes the
probability with a fixed sign. The resulting nonzero boundary
derivatives project onto coefficient variation from the centered lag
design. This gives a positive variance bound away from the mean lag
state, without sign symmetry of the innovation copula.

For nonempty $I\subseteq\{1,\ldots,n\}$, signs $s_i\in\{-1,1\}$,
and finite thresholds $c_i$, consider
\begin{equation}
 A=\{y:s_i y_i\le s_i c_i\text{ for every }i\in I\}.
 \label{gnd:event}
\end{equation}
Other coordinates are unrestricted. Continuity makes strict and weak
boundaries equivalent in probability. Write
$\mu_x=(1,\mu_L')'$ and $\Sigma_L=\operatorname{Var}(L_t)$.
For $r=(1,l')'$, let $f_i(r)$ be the signed-face derivative in
\eqref{gnd:positive-face}, and set $k_i(r)=-s_if_i(r)$ on $I$ and
zero elsewhere.

\begin{proposition}[One-sided event nondegeneracy]\label{ond:proposition}
Under Theorem~\ref{main:dynamic}, suppose the innovation vine is
Gaussian or positive Clayton as in
Propositions~\ref{score:gaussian-family} and~\ref{tail:sequential-clayton}.
Then $f_i(r)>0$ for $i\in I$, and
\begin{equation}
 d_A(r)'\mathcal W d_A(r)
 \ge (l-\mu_L)'\Sigma_L^{-1}(l-\mu_L)\,k(r)'R_Zk(r)>0
 \quad\text{if }r\ne\mu_x,
 \label{ond:lower-bound}
\end{equation}
for $\mathcal W=\mathcal W_{\rm VAR}$ or $\mathcal W_T^{\rm fin}$,
$1\le b<T$. Since $P(X=\mu_x)=0$, the projected variance is positive
almost surely.
\end{proposition}

\begin{proof}
The coefficient influence contains variation in the lag design that
cannot be explained by functions of the current innovation. To isolate
it, use $Qe_1=\mu_x$ and write
\begin{equation}
 w_{i,t}=e_1Z_{it}+u_{i,t},\qquad
 u_{i,t}=Q^{-1}(x_t-\mu_x)Z_{it},\qquad
 K_x=Q^{-1}-e_1e_1'.
 \label{ond:design-decomposition}
\end{equation}
Current $Z_t$ is independent of $x_t$. For any scalar square-integrable
function $F(Z_t)$, Cauchy--Schwarz ensures that $E|Z_{it}F(Z_t)|<\infty$,
and therefore
\begin{align}
 E\{u_{i,t}F(Z_t)\}
 &=Q^{-1}E(x_t-\mu_x)E\{Z_{it}F(Z_t)\}=0,\notag\\
 E(u_{i,t}u_{j,t}')
 &=(R_Z)_{ij}Q^{-1}(Q-\mu_x\mu_x')Q^{-1}
   =(R_Z)_{ij}K_x.
 \label{ond:orthogonality}
\end{align}
All score and scale coordinates used below are square integrable by
the polynomial envelopes and normal coordinate moments. Stack the
$u_{i,t}$ by equation and put $\upsilon_t=(0_{q_\theta}',u_t',0_n')'$.
The common influence $\xi_t$ in \eqref{dji:ideal-xi} decomposes as
$\xi_t=\upsilon_t+\chi(Z_t)$: replace each coefficient block by
$e_1Z_{it}$ to obtain $\chi$. Equation~\eqref{ond:orthogonality} makes
the two terms orthogonal. The positive semidefinite decomposition \eqref{dji:psd-identity}
then implies
\begin{equation}
 d_A'\mathcal W_{\rm VAR}d_A
 \ge d_{A,\beta}'(R_Z\otimes K_x)d_{A,\beta}.
 \label{ond:variance-bound}
\end{equation}
For \eqref{fdw:population-congruence}, write each weighted influence as
$L_{T,t}\zeta_t=\upsilon_t+\chi_{T,t}(Z_t)$.
Equation~\eqref{ond:orthogonality} gives
\[
 E\{\upsilon_t\chi_{T,t}(Z_t)'\}=0,\qquad
 \mathcal W_T^{\rm fin}
 =E(\upsilon_t\upsilon_t')
       +T^{-1}\sum_{t=1}^TE\{\chi_{T,t}(Z_t)\chi_{T,t}(Z_t)'\}
 \succeq E(\upsilon_t\upsilon_t').
\]
The coefficient block of the last matrix is $R_Z\otimes K_x$ and
its other blocks are zero. This proves \eqref{ond:variance-bound}
also for the finite array.

The sign and size of the coefficient gradient follow from the event
boundaries. Extend $s_i$ to all coordinates by setting $s_i=1$ outside
$I$, put $S_I^*=\operatorname{diag}(s_1,\ldots,s_n)$, and let
$U=S_I^*Z$ have density $f_U$. For $i\in I$, write
$t_i(r)=s_i(c_i-r'\beta_i)/\sigma_i$ and
$D_{-i}(r)=\{u_{-i}:u_j\le t_j(r),\ j\in I\setminus\{i\}\}$.
The density is continuous and strictly positive everywhere for either
verified family. The face differentiation and finite trace bounds of
Lemma~\ref{dfe:rectangle-smooth} give
\begin{equation}
 f_i(r)=\int_{D_{-i}(r)}f_U\{t_i(r),u_{-i}\}\,du_{-i}>0.
 \label{gnd:positive-face}
\end{equation}
For $n>1$, the domain contains a compact box of positive volume: use
$[t_j(r)-1,t_j(r)-1/2]$ for restricted coordinates and $[-1,1]$ for
the others. The continuous positive density has a positive minimum
on this box, so its integral is strictly positive. For $n=1$, the
integral is the positive univariate density. A normalized coefficient
perturbation $d_i$ changes the signed boundary by $-s_i r'd_i$.
Consequently
\begin{equation}
 d_{A,\beta_i}(r)=k_i(r)r,\qquad
 d_{A,\beta}'(R_Z\otimes K_x)d_{A,\beta}
       =\{k(r)'R_Zk(r)\}\{r'K_xr\}.
 \label{gnd:coefficient-gradient}
\end{equation}
At least one $k_i(r)$ is nonzero, and $R_Z>0$, so the first factor
is positive.

To evaluate the second factor, block inversion of
$Q=\left(\begin{smallmatrix}1&\mu_L'\\
\mu_L&\Sigma_L+\mu_L\mu_L'\end{smallmatrix}\right)$ gives
\[
 K_x=\begin{pmatrix}
 \mu_L'\Sigma_L^{-1}\mu_L&-\mu_L'\Sigma_L^{-1}\\
 -\Sigma_L^{-1}\mu_L&\Sigma_L^{-1}
 \end{pmatrix},\qquad
 r'K_xr=(l-\mu_L)'\Sigma_L^{-1}(l-\mu_L).
\]
The matrix $\Sigma_L$ is positive definite because it is the Schur
complement of the intercept entry in $Q>0$. The last display is
strictly positive when $l\ne\mu_L$, proving \eqref{ond:lower-bound}.
Finally, conditional on $\mathcal F_{t-2}$, the first coordinate of
$y_{t-1}$ has a normal distribution with variance $\sigma_1^2>0$.
It therefore has no atom at its unconditional mean. Since $p\ge1$
and this coordinate belongs to $L_t$,
\[
 P(x_t=\mu_x)\le P\{y_{1,t-1}=E(y_{1,t-1})\}
 =E[P\{y_{1,t-1}=E(y_{1,t-1})\mid\mathcal F_{t-2}\}]=0.
\]
Stationarity transfers this conclusion to $X$. The bound proves
almost-sure nondegeneracy, without a uniform lower bound over states
and without asserting positivity at $r=\mu_x$ in the Clayton case.
\end{proof}

\subsection{Intervals on the log-odds scale}

A logit interval requires control near probability boundaries because
the random forecast state can make the true probability arbitrarily
close to zero or one. We first restrict probabilities away from the
boundaries and then let that restriction shrink. This justifies the
logit derivative without a uniform lower bound over states.

Let $\widehat p_T$ estimate a fixed rectangular-event probability
$p(x_{T+1})$, and let $\widehat\tau_T$ be the corresponding nonnegative
projected standard deviation. Write
$\ell(p)=\log\{p/(1-p)\}$ and $\Lambda(u)=(1+e^{-u})^{-1}$.
For deterministic $\epsilon_T\in(0,1/2)$ with $\epsilon_T\to0$, define
\begin{align}
 \check p_T&=\min(1-\epsilon_T,\max(\epsilon_T,\widehat p_T)),\notag\\
 \widehat\sigma_{\ell,T}&=
         \frac{\widehat\tau_T}{\check p_T(1-\check p_T)},\notag\\
 I_{\ell,T}&=\left[
 \Lambda\left\{\ell(\check p_T)-z_{1-\alpha/2}
                    \widehat\sigma_{\ell,T}/\sqrt T\right\},
 \Lambda\left\{\ell(\check p_T)+z_{1-\alpha/2}
                    \widehat\sigma_{\ell,T}/\sqrt T\right\}
                    \right].
 \label{rli:logit-interval}
\end{align}
The construction applies separately to the archive, its finite-weight
covariance, and the Gaussian refit.

\begin{proposition}[Logit event intervals]\label{rli:logit}
Under the event-inference conditions of Theorem~\ref{dfe:event-clt},
Corollary~\ref{dfe:rectangle-clt}, or Corollary~\ref{dfs:event-inference},
suppose $P\{0<p(X)<1,\ \tau^2(X)>0\}=1$. For each fixed
$0<\alpha<1$,
\begin{equation}
 \frac{\sqrt T\{\ell(\check p_T)-\ell(p(x_{T+1}))\}}
      {\widehat\sigma_{\ell,T}}
 \Longrightarrow N(0,1),\qquad
 P\{p(x_{T+1})\in I_{\ell,T}\}\longrightarrow1-\alpha.
 \label{rli:logit-pivot}
\end{equation}
The pivot is set to zero when its denominator is zero, an event of
vanishing probability.
\end{proposition}

\begin{proof}
Put $P_T=p(x_{T+1})$ and $E_T=\widehat p_T-P_T$. The event theorem
supplies $\sqrt T E_T=O_p(1)$,
$\widehat\tau_T^2-\tau^2(x_{T+1})=o_p(1)$, and the joint state limit.
For $0<\delta<1/2$, stationarity gives
\begin{equation}
 P\{P_T\notin[\delta,1-\delta]\}
 =P\{p(X)\notin[\delta,1-\delta]\}\longrightarrow0
 \quad\text{as }\delta\downarrow0.
 \label{rli:boundary-localization}
\end{equation}
On $\mathcal B_{T,\delta}=\{P_T\in[\delta,1-\delta],|E_T|\le\delta/2\}$,
clipping is inactive for $\epsilon_T<\delta/2$ and
\[
 P(\mathcal B_{T,\delta}^c)
 \le P\{p(X)\notin[\delta,1-\delta]\}+P(|E_T|>\delta/2).
\]
Taking first $T\to\infty$ and then $\delta\downarrow0$ gives
$P(\check p_T\ne\widehat p_T)\to0$ without a uniform probability bound.

On the same event,
$\ell'(p)=\{p(1-p)\}^{-1}$ and
$\ell''(p)=(2p-1)/\{p^2(1-p)^2\}$ have finite bounds depending only
on $\delta$. Taylor's integral formula gives
\begin{align}
 \ell(\widehat p_T)-\ell(P_T)
 &=\frac{E_T}{P_T(1-P_T)}+R_T,\notag\\
 |R_T|\mathbf1_{\mathcal B_{T,\delta}}
 &\le C_\delta E_T^2,\qquad
 \sqrt T|R_T|\mathbf1_{\mathcal B_{T,\delta}}
       \le C_\delta(\sqrt T E_T)^2/\sqrt T=o_p(1).
 \label{rli:logit-remainder}
\end{align}
Define $R_T$ globally by the first equality with $\check p_T$ in
place of $\widehat p_T$. The complement bound and
\eqref{rli:boundary-localization} extend the $o_p(1)$ conclusion
unconditionally. On $\mathcal B_{T,\delta}$,
\[
 \left|\frac{\check p_T(1-\check p_T)}{P_T(1-P_T)}-1\right|
 \le C_\delta |E_T|=o_p(1).
\]
The same localization therefore makes this ratio converge to one in
probability without a global lower bound on $P_T(1-P_T)$.

Positivity and the variance localization in
\eqref{dfe:small-variance} give
$\widehat\tau_T/\tau(x_{T+1})\to_p1$ and
$1/\widehat\tau_T=O_p(1)$, defining this inverse as zero when
$\widehat\tau_T=0$. Thus dividing the Taylor expansion yields
\begin{align}
 \frac{\sqrt T\{\ell(\check p_T)-\ell(P_T)\}}
                    {\widehat\sigma_{\ell,T}}
 &=\frac{\sqrt T E_T}{\widehat\tau_T}
       \frac{\check p_T(1-\check p_T)}{P_T(1-P_T)}
   +\frac{\sqrt T R_T\check p_T(1-\check p_T)}{\widehat\tau_T}
       \notag\\
 &=\frac{\sqrt T E_T}{\widehat\tau_T}\{1+o_p(1)\}+o_p(1).
 \label{rli:pivot-transfer}
\end{align}
Here $\check p_T(1-\check p_T)\le1/4$ and
$\sqrt T R_T=o_p(1)$. The event pivot and Slutsky's theorem imply
\eqref{rli:logit-pivot}. Monotonicity of $\Lambda$ identifies the
coverage event with the pivot lying between the normal quantiles.
Their limiting boundary probabilities are zero.

For a nontrivial fixed rectangle and a strictly positive innovation
density on $\mathbb R^n$, $0<p(r)<1$ for every finite state: its
transformed rectangle contains an open box of positive density, and
its complement contains another whenever at least one boundary is
finite. This establishes the interior-probability condition for the
stated Gaussian and positive-parameter Clayton families. For one-sided events, Proposition~\ref{ond:proposition} establishes
almost-sure positive archive variance in both families, and
Proposition~\ref{gnd:proposition} covers the Gaussian refit.
General bounded rectangles retain their explicit nondegeneracy assumption.
\end{proof}

Finite estimated standard deviations give interval endpoints in
$(0,1)$. The result establishes first-order coverage.

\section{Gaussian Refitting and the Archive Variance Cost}\label{supp:dynamic-refit}\label{app:refit}
The terminal refit recomputes historical residuals from the same
observations as the archive. Quadratic Gaussian edge scores allow an
exact OLS residual cross-product identity to determine its
generated-margin correction, accounting for the dependence of lag
designs on past innovations. Comparing the resulting joint influence
with the archive influence gives the positive semidefinite covariance
difference at common population parameters.

For the Gaussian full-sample fit \eqref{dfs:estimator}, define
\begin{equation}
 \xi_t^{\rm fs}=\begin{pmatrix}
       -J^{-1}\{g(\theta_0,Z_t)+B_g\varphi_t\}\\e_t
       \end{pmatrix},\qquad
 \mathcal W_{\rm fs}=E(\xi_t^{\rm fs}\xi_t^{{\rm fs}\prime}).
 \label{dfs:influence}
\end{equation}

\begin{proposition}[Gaussian refit limit and covariance]\label{dfs:theorem}
Under Assumptions~\ref{main:var-assumption}--\ref{main:regularity}
with a Gaussian innovation vine,
\begin{equation}
 \sqrt T\begin{pmatrix}\widetilde\theta-\theta_0\\\delta_T\end{pmatrix}
 =\frac1{\sqrt T}\sum_{t=1}^T\xi_t^{\rm fs}+o_p(1)
 \ \Longrightarrow N(0,\mathcal W_{\rm fs}).
 \label{dfs:limit}
\end{equation}
The estimator \eqref{dfs:feasible}, with its stated identity-matrix convention,
is positive semidefinite for every sample and consistent for
$\mathcal W_{\rm fs}$. Moreover,
\begin{equation}
 \mathcal W_{\rm VAR}-\mathcal W_{\rm fs}
 =\begin{pmatrix}
 J^{-1}B_g\Sigma_\varphi B_g'J^{-\prime}&0\\0&0
 \end{pmatrix}\succeq0.
 \label{dfs:cost}
\end{equation}
\end{proposition}

\begin{proof}
Online Proposition~\ref{refit:theorem}, with $b=0$, gives the joint
expansion and limit. Its feasible covariance
\eqref{refit:finite-feasible} reduces to \eqref{dfs:feasible}.
The following direct Gaussian calculation identifies the source of
the marginal correction.

Let $D_T=(d_{1,T},\ldots,d_{n,T})$,
$M_T=T^{-1}\sum_tZ_tZ_t'$, $R=E(Z_tZ_t')$, and
$L_T=\operatorname{diag}\{(1+a^v_{i,T})^{-1/2}:i\le n\}$.
On the bounded-inverse and inactive-safeguard event of
Lemma~\ref{dji:endpoint-lemma}, the OLS normal equations give
\begin{equation}\label{dfs:normal-equations}
 \widetilde M_T:=T^{-1}\sum_t\widetilde Z_t\widetilde Z_t'
 =L_T(M_T-D_T'\widehat Q_TD_T)L_T.
\end{equation}
Here $D_T=O_p(T^{-1/2})$, $M_T-R=O_p(T^{-1/2})$, and
$\|\widehat Q_T\|=O_p(1)$. Put
$V_T=\operatorname{diag}(a^v_{1,T},\ldots,a^v_{n,T})$.
Taylor expansion gives $L_T=I-V_T/2+O_p(T^{-1})$. The omitted
terms in the residual-moment expansion are bounded by
\[
 C\bigl\{\|D_T\|^2\|\widehat Q_T\|+
         \|V_T\|\|M_T-R\|+\|V_T\|^2(1+\|M_T\|)\bigr\}
 =O_p(T^{-1}).
\]
Consequently
\begin{equation}\label{dfs:moment-expansion}
 \widetilde M_T=M_T-\tfrac12(V_TR+RV_T)+O_p(T^{-1}).
\end{equation}

The Gaussian conditional-score recursion is linear in the original
normal scores. Each edge score therefore has symmetric $H_e$ such that
\begin{equation}\label{dfs:quadratic}
 g_e(\theta_0,z)=c_e+z'H_ez,\qquad
 (B_g)_{e,\mu}=0,\qquad (B_g)_{e,v_i}=-(RH_e)_{ii}.
\end{equation}
Indeed, the location derivative is $-2(H_ez)_i$ and has expectation
zero. The relative-variance derivative is $-z_i(H_ez)_i$, whose
expectation is $-(RH_e)_{ii}$. Taking the trace of
\eqref{dfs:moment-expansion} against $H_e$ gives
\[
 T^{-1}\sum_t\{g_e(\theta_0,\widetilde Z_t)-g_e(\theta_0,Z_t)\}
 =-\sum_i(RH_e)_{ii}a^v_{i,T}+O_p(T^{-1}).
\]
Substitute the terminal scale expansion in
Lemma~\ref{dji:endpoint-lemma} to obtain the Gaussian version of
\eqref{refit:score-leading}. The intercept and lag-coefficient errors
enter the residual cross product only through the quadratic term in
\eqref{dfs:normal-equations}. This explains why historical location
errors make no first-order contribution in this specialization.

For later numerical calculations, the terminal residual bound is
\begin{equation}\label{dfs:residual-bound}
 \|\widetilde Z_t-Z_t\|
 \le Cr_T(\|x_t\|+\|Z_t\|),\qquad
 r_T=\|D_T\|+\max_i|a^v_{i,T}|=O_p(T^{-1/2}).
\end{equation}
It follows directly from
$\widetilde Z_{it}=(Z_{it}-x_t'd_{i,T})(1+a^v_{i,T})^{-1/2}$
on the same high-probability event. Finally,
\eqref{main:covariance-decomposition} and \eqref{dfs:influence}
give \eqref{dfs:cost}. For each copula direction $a$, the additional
quadratic form is
$\|\Sigma_\varphi^{1/2}B_g'J^{-\prime}a\|^2\ge0$.
\end{proof}

For a fixed rectangle $A$ as in Lemma~\ref{dfe:rectangle-smooth},
let $\widetilde p_{A,T+1}$ and $\widetilde d_{A,T}$ substitute
$(\widetilde\theta,\widehat\beta_T,\widehat v_T)$ in the probability
and gradient formulas.

\begin{corollary}[Gaussian refit event inference]\label{dfs:event-inference}
Under Proposition~\ref{dfs:theorem}, if
$P\{d_A(X)'\mathcal W_{\rm fs}d_A(X)>0\}=1$, then
\begin{equation}
 \frac{\sqrt T(\widetilde p_{A,T+1}-p_{A,T+1})}
 {\{\widetilde d_{A,T}(x_{T+1})'
       \widehat{\mathcal W}_{\rm fs}
       \widetilde d_{A,T}(x_{T+1})\}^{1/2}}
 \ \Longrightarrow N(0,1).
 \label{dfs:event}
\end{equation}
The associated normal interval, intersected with $[0,1]$, has
limiting coverage $1-\alpha$ for each fixed $0<\alpha<1$.
The pivot is set to zero when its denominator is zero.
\end{corollary}

\begin{proof}
Take $b=0$ in Online Corollary~\ref{refit:event-inference}.
Its estimator and covariance are exactly
\eqref{dfs:estimator} and \eqref{dfs:feasible}. The imposed variance
condition is the one in that corollary, so its separation and
studentization argument gives \eqref{dfs:event}. Normal quantiles
give coverage $1-\alpha$. Since the target lies in $[0,1]$,
intersecting the interval with $[0,1]$ preserves its coverage event.
\end{proof}

\endgroup
\clearpage
\begingroup
\setlength{\bibsep}{5pt}
\interlinepenalty=10000
\bibliographystyle{jfe}
\bibliography{VAR_ref}
\endgroup
\end{document}